\documentclass[onecolumn,12pt,draftcls]{IEEEtran}
\IEEEoverridecommandlockouts

\usepackage{amsmath}
\usepackage{amssymb}
\usepackage{amsthm}
\usepackage{dsfont}
\usepackage{enumitem}
\usepackage{bm}
\usepackage{hyperref}
\usepackage{cite}
\usepackage[normalem]{ulem}

\usepackage{graphicx}
\usepackage{xcolor}

\newcommand{\ind}[1]{\mathds{1}\{#1\}}
\newcommand{\E}[2]{\mathbb{E} _{ #1 }  \left[ #2 \right]}
\newcommand{\Var}[2]{\mathrm{Var} _{#1} \left( #2 \right)}
\newcommand{\Prob}[2]{\mathbb{P} _{ #1 } \left\{ #2 \right\}}
\newcommand{\Pb}{\mathbb{P}}

\newcommand{\floor}[1]{\left\lfloor #1 \right\rfloor}
\newcommand{\norm}[1]{\left\lVert #1 \right\rVert}
\newcommand{\abs}[1]{\left| #1 \right|}
\newcommand{\R}{\mathbb{R}}
\newcommand{\gfun}{g}

\newcommand{\T}{\intercal}
\newcommand{\td}{\Tilde}
\newcommand{\KL}[2]{D(#1||#2)}
\newcommand{\emax}[1]{\lambda_{\text{max}} \left( #1 \right)}
\newcommand{\eps}{\varepsilon}
\newcommand{\D}{\mathrm{d}}
\newcommand{\const}{\mathsf{const}}
\newcommand{\Ccal}{ \mathcal{C}_{\alpha} }
\newcommand{\Cbb}{\mathbb{C}}

\newcommand{\FAR}[1]{\mathrm{FAR} \left( #1 \right)}
\newcommand{\SPFA}{\mathrm{SPFA}}
\newcommand{\ADD}{\mathrm{ADD}}
\newcommand{\WADD}[1]{\mathrm{WADD} \left( #1 \right)}
\newcommand{\SADD}[1]{\mathrm{SADD} \left( #1 \right)}
\newcommand{\WADDth}[1]{\mathrm{WADD} _{ \theta } \left( #1 \right)}
\newcommand{\SADDth}[1]{\mathrm{SADD} _{ \theta } \left( #1 \right)}

\DeclareMathOperator*{\esssup}{ess\,sup}

\theoremstyle{plain}
\newtheorem{theorem}{Theorem}[section]
\newtheorem{lemma}{Lemma}[section]

\theoremstyle{definition}

\newtheorem{example}{Example}[section]
\newtheorem{assumption}{Assumption}[section]

\theoremstyle{remark}
\newtheorem*{remark}{Remark}    

\newcommand{\vvv}[1]{#1} 
\newcommand{\ignore}[1]{} 

\begin{document}


\title{Quickest Change Detection with Non-Stationary Post-Change Observations}

\author{Yuchen Liang, ~\IEEEmembership{Graduate Student Member,~IEEE},\thanks{Y. Liang and V.V.~Veeravalli are with the ECE Department and Coordinated Science Laboratory, 
University of Illinois at Urbana-Champaign; email: yliang35,vvv@ILLINOIS.EDU}  Alexander G. Tartakovsky, ~\IEEEmembership{Senior Member,~IEEE},
\thanks{A. G. Tartakovsky is President of AGT StatConsult, Los Angeles, California, USA; e-mail: alexg.tartakovsky@gmail.com} 
and  Venugopal V. Veeravalli, ~\IEEEmembership{Fellow, ~IEEE}\thanks{This work was supported in part by the National Science Foundation under grant ECCS-2033900, and by the Army Research Laboratory under Cooperative Agreement W911NF-17-2-0196, through the University of Illinois at Urbana-Champaign.}
}

\maketitle

\vspace*{-0.5in}

\begin{abstract}
The problem of quickest detection of a change in the distribution of a sequence of independent observations is considered. The pre-change observations are assumed to be stationary with a known distribution, while the post-change observations are allowed to be non-stationary with some possible parametric uncertainty in their distribution. In particular, it is assumed that the cumulative Kullback-Leibler divergence between the post-change and the pre-change distributions grows in a certain manner with time after the change-point. For the case where the post-change distributions are known, a universal asymptotic lower bound on the delay is derived, as the false alarm rate goes to zero. Furthermore, a window-limited Cumulative Sum (CuSum) procedure is developed, and shown to achieve the lower bound asymptotically. For the case where the post-change distributions have parametric uncertainty, a window-limited (WL) generalized likelihood-ratio (GLR) CuSum procedure is developed and is shown to achieve the universal lower bound asymptotically. Extensions to the case with dependent observations are discussed. The analysis is validated through numerical results on synthetic data. The use of the WL-GLR-CuSum procedure in monitoring pandemics is also demonstrated.
\end{abstract}

\begin{IEEEkeywords}
Quickest change detection, non-stationary observations, CuSum procedure, generalized likelihood-ratio CuSum procedure.
\end{IEEEkeywords}

\section{Introduction}




The problem of quickest change detection (QCD) is of fundamental importance in a variety of applications and has been extensively studied in mathematical statistics (see, e.g., \cite{tartakovsky_sequential,tartakovsky_qcd2020,vvv_qcd_overview,xie_vvv_qcd_overview} for overviews). Given a sequence of observations whose distribution changes at some unknown change point, the goal is to detect the change in distribution as quickly as possible after it occurs, while not making too many false alarms. 

In the classical formulations of the QCD problem, it is assumed that observations are independent and identically distributed (i.i.d.) with known pre- and 
post-change distributions. In many practical situations, while it is reasonable to assume that we can accurately estimate the pre-change distribution, the post-change distribution is rarely completely known. Furthermore, in many cases, it is reasonable to assume that the system is in a steady state before the change point and produces i.i.d.\ observations, but in the post-change mode the observations may be substantially non-identically distributed, i.e., non-stationary. 
For example, in the pandemic monitoring problem, the distribution of the number of people infected daily might have achieved a steady (stationary) state before the start of a new wave, but after the onset of a new wave, the post-change observations may no longer be stationary. Indeed, during the early phase of the new wave, the mean of the post-change distribution grows approximately exponentially. We will address the pandemic monitoring problem in detail in Section~\ref{sec:num-res}.

In this paper, our main focus is on the QCD problem with independent observations\footnote{The extension to the case of dependent observations is discussed in Section~\ref{sec:ext}.}, where the pre-change observations are assumed to be stationary with a known distribution, while the post-change observations are allowed to be non-stationary with some possible parametric uncertainty in their distribution.

There have been extensions of the classical formulation to the case where the pre- and/or post-change distributions are not fully known and observations
may be non-i.i.d., i.e., dependent and nonidentically distributed.  For the i.i.d. case with parametric uncertainty in the post-change regime, Lorden~\cite{lorden1971} proposed a generalized likelihood ratio (GLR) Cumulative Sum (CuSum) procedure, and proved its
asymptotic optimality in the minimax sense as the false alarm rate goes to zero, for one-parameter exponential families. An alternative to the GLR-CuSum, the mixture-based CuSum, was proposed and studied by Pollak \cite{pollakmixture} in the same setting as in \cite{lorden1971}.
The GLR approach has been studied in detail for the problem of detecting the change in the mean of a Gaussian i.i.d. sequence with an unknown post-change mean by Siegmund~\cite{siegmund1995}.
\vvv{Both the mixture-based and GLR-CuSum procedures have been studied by Lai~\cite{lai1998} in the pointwise setting in the non-i.i.d. case of possibly dependent and non-identically distributed observations, with parametric uncertainty in the post-change regime. More specifically, in \cite{lai1998}, Lai assumed that the log-likelihood ratio process (between post- and pre-change distributions) normalized by the number of observations $n$ converges to a positive and finite constant as $n \to\infty$, which can be interpreted as a Kullback-Leibler (KL) information number. In the case of independent (but non-identically distributed observations) this means that the expected value of the log-likelihood ratio process grows approximately linearly in the number of observations $n$, for large 
$n$. Tartakovsky~\cite{TartakovskySISP98} and Tartakovsky et al.~\cite{tartakovsky_sequential} refer to such a case as ``asymptotically homogeneous'' (or stationary) case. In \cite{lai1998}, Lai also developed a universal lower bound on the worst-case expected delay as well as on the expected delay to detection for every change point and proved that 
a specially designed window-limited (WL) CuSum procedure asymptotically achieves the lower bound as the maximal probability of false alarm approaches 0, when both pre- and post-change distributions are completely
known, i.e., that the designed WL-CuSum procedure is asymptotically pointwise optimal to first order. For the case where the post-change distribution has parametric uncertainty, Lai proposed and analyzed a WL-GLR-CuSum procedure. A general Bayesian theory for non-i.i.d. asymptotically stationary stochastic models has been developed
by Tartakovsky and Veeravalli~\cite{TartakovskyVeerTVP05} and Tartakovsky~\cite{TartakovskyIEEEIT2017} for the discrete-time scenario, and by Baron and Tartakovsky~\cite{BaronTartakovskySA06} for the continuous-time scenario, when both pre- and post-change models are completely known. It was shown in these works that a Shiryaev-type change detection procedure minimizes not only average detection delay but also higher moments of the detection delay asymptotically, as the weighted probability of false alarm goes to zero, under very general conditions for the prior distribution of the change point. Extensions of these results to the case of the parametric composite post-change hypothesis have been provided by Tartakovsky~\cite{TartakovskyIEEEIT2019,tartakovsky_qcd2020} where it has been shown that mixture
Shiryaev-type detection rule is asymptotically first-order optimal in the Bayesian setup and by Pergamenchtchikov and Tartakovsky~\cite{PerTar-JMVA2019}
where it was shown that the mixture Shiryaev-Roberts-type procedure pointwise and minimax asymptotically optimal in the non-Bayesian setup, but in the asymptotically stationary case where the cumulative KL divergence between post- and pre-change distributions $g(n)$ grows linearly in $n$ as $n\to\infty$. }

\vvv{\emph{Note that all the previously cited works focus on the asymptotically stationary case.}} To the best of our knowledge, the asymptotically non-stationary case where the 
expected value of the log-likelihood ratio process normalized to some nonlinear function $g(n)$ converges to a positive and finite (information) number has never been 
considered.\footnote{It should be noted that such an asymptotically non-stationary case has been previously considered for sequential hypothesis testing problems by Tartakovsky~\cite{TartakovskySISP98} and 
Tartakovsky et al.~\cite{tartakovsky_sequential} .}
Our contributions are as follows:
\begin{enumerate}
  \item We develop a universal asymptotic (as the false alarm rate goes to zero) lower bound on the worst-case expected delay for our problem setting with non-stationary post-change observations.
   
  \item We develop a window limited CuSum (WL-CuSum) procedure that asymptotically achieves the lower bound on the worst-case expected delay when the post-change distribution is fully known.
   
  \item We develop and analyze a WL-GLR-CuSum procedure that asymptotically achieves the worst-case expected delay when the post-change distribution has parametric uncertainty.
   
  \item We validate our analysis through numerical results and demonstrate the use of our approach in monitoring pandemics.
\end{enumerate}

The rest of the paper is structured as follows. In Section~\ref{sec:info-bd}, we derive the information bounds and propose an asymptotically optimal WL-CuSum procedure when the post-change distribution completely known. In Section~\ref{sec:unknown-param}, we propose an asymptotically optimal WL-GLR-CuSum procedure when the post-change distribution has unknown parameters.
In Section~\ref{sec:ext}, we discuss possible extensions to the general non-i.i.d. case where the observations can be dependent and non-stationary. 
In Section~\ref{sec:num-res}, we present some numerical results, including results on monitoring pandemics. We conclude the paper in Section~\ref{sec:concl}. In the Appendix, we provide proofs of certain results.

\section{Information Bounds and Optimal Detection}
\label{sec:info-bd}

Let $\{X_n\}_{n\ge 1}$ be a sequence of independent random variables (generally vectors), and let $\nu$ be a change point. 
Assume that $X_1, \dots, X_{\nu-1}$ all have density $p_0$ with respect to some non-degenerate, sigma-finite measure $\mu$ and 
that $X_\nu, X_{\nu+1}, \dots$ have densities $p_{1,\nu,\nu}, p_{1,\nu+1,\nu}, \ldots$, respectively, with respect to $\mu$. 
Note that the observations are allowed to be non-stationary after the change point and the post-change distributions may generally depend on the change point.

Let $({\cal F}_{n})_{n\ge 0}$ be the filtration, i.e.,  ${\cal F}_{0}=\{\Omega,\varnothing\}$ and ${\cal F}_{n}=\sigma\left\{X_{\ell}, 1\le \ell \le n \right\}$ 
is the sigma-algebra generated by the vector of $n$ observations $X_1,\dots,X_n$ and let ${\cal F}_\infty= \sigma(X_1,X_2, \dots)$. 
In what follows we denote by $\mathbb{P}_\nu$ the probability measure on the entire sequence of observations when the change-point is $\nu$. That is, 
under $\mathbb{P}_\nu$ the random variables $X_1, \dots, X_{\nu-1}$ are i.i.d. with the common (pre-change) density $p_0$ 
and $X_\nu, X_{\nu+1}, \dots$ are independent with (post-change) densities  $p_{1,\nu,\nu}, p_{1,\nu+1,\nu}, \ldots$ .
Let $\mathbb{E}_\nu$ denote the corresponding expectation. For $\nu=\infty$ this distribution will be denoted by $\mathbb{P}_\infty$ and the corresponding expectation by $\mathbb{E}_\infty$. Evidently, under $\mathbb{P}_\infty$ the random variables  $X_1,X_2,\dots$ are i.i.d. with density $p_0$.  In the sequel, we denote by $\tau$ Markov (stopping) times with respect to the filtration $({\cal F}_{n})_{n\ge 0}$, i.e., the event $\{\tau=n\}$ belongs to ${\cal F}_n$. 

The change-time $\nu$ is assumed to be unknown but deterministic. 
The problem is to detect the change quickly while not causing too many false alarms. Let $\tau$ be a stopping time defined on the observation sequence associated with the detection rule, i.e., $\tau$ is the time at which we stop taking observations and declare that the change has occurred. The problem is to detect the change quickly, minimizing the delay to detection $\tau -\nu$, while not causing too many false alarms. 

\subsection{Classical Results under i.i.d. Model}
\label{subsec:stat_postc}

A special case of the model described above is where both the pre- and post-change observations are i.i.d., i.e., $p_{1,n,\nu} \equiv p_1$ for all $n \geq \nu \geq 1$. In this case, Lorden \cite{lorden1971} proposed solving the following optimization problem to find the best stopping time $\tau$:
\begin{equation}
\label{prob_def}
    \inf_{\tau \in \mathcal{C}_\alpha} \WADD{\tau}
\end{equation}
where
\begin{equation}
\label{LordenADD}
    \WADD{\tau} := \sup_{\nu \geq 1} \esssup \E{\nu}{\left(\tau-\nu+1\right)^+|{\cal F}_{\nu-1}}
\end{equation}
characterizes the worst-case expected delay, and $\esssup$ stands for essential supremum. The constraint set is
\begin{equation}
\label{fa_constraint}
    \mathcal{C}_\alpha := \left\{ \tau: \FAR{\tau} \leq \alpha \right\}
\end{equation}
with  
\begin{equation}
    \FAR{\tau} := \frac{1}{ \E{\infty}{\tau}} \label{eq:FAR_def}
\end{equation}   
which guarantees that the false alarm rate of the algorithm does not exceed $\alpha$. Recall that $\E{\infty}{\cdot}$ is the expectation operator when the change never happens, and we use the conventional notation $(\cdot)^+:=\max\{0,\cdot\}$ for the nonnegative part. The mean time to a false alarm (MTFA) 
$\E{\infty}{\tau}$  is sometimes referred to as the average run length to false alarm.

Lorden also showed that Page's CuSum detection algorithm \cite{page1954}
solves the problem in \eqref{prob_def} asymptotically as $\alpha \to 0$,
which is given by:
\begin{equation}
\label{defstoppingrule}
    \tau_{\text{Page}}\left(b\right) := \inf \left\{n: \max_{1\leq k \leq n+1} \sum_{i=k}^n \log \frac{p_1(X_n)}{p_0(X_n)} \geq b \right\}.
\end{equation}
It was shown by Moustakides \cite{moustakides1986} that the CuSum algorithm is exactly optimal for the problem in (\ref{prob_def}) if threshold $b=b_\alpha$ 
is selected so that $\FAR{\tau_{\text{Page}}\left(b_\alpha\right)} = \alpha$. If threshold $b_\alpha$ is selected in a special way that accounts for 
the overshoot of $W(n)$ over $b_\alpha$ at stopping, which guarantees the approximation $\FAR{\tau_{\text{Page}}\left(b_\alpha\right)} \sim \alpha$ 
as $\alpha\to0$, then we have the following third-order asymptotic approximation (as $\alpha\to 0$) for the worst-case expected detection delay of the optimal procedure:
\begin{align*}
 \inf_{\tau \in \mathcal{C}_\alpha} \WADD{\tau} & = \WADD{\tau_{\text{Page}}(b_\alpha)} + o(1),
 \\
\WADD{\tau_{\text{Page}}(b_\alpha)} & = \frac{1}{{\KL{p_1}{p_0}}}
 (\abs{\log \alpha} - \const + o(1))
\end{align*}
(see, e.g., \cite{tartakovsky_sequential}), which also implies the first-order asymptotic approximation (as $\alpha \to 0$):
\begin{equation} \label{FOWADDPage}
    \inf_{\tau \in \mathcal{C}_\alpha} \WADD{\tau} \sim \WADD{\tau_{\text{Page}}\left(\abs{\log\alpha}\right)} = \frac{\abs{\log \alpha}}{\KL{p_1}{p_0}} (1+o(1))
\end{equation} 
where $Y_\alpha\sim G_\alpha$ is equivalent to $Y_\alpha = G_\alpha (1+o(1))$. Here $\KL{p_1}{p_0}$ is the Kullback-Leibler (KL) divergence between $p_1$ and $p_0$. Also, in the following we use a standard notation
$o(x)$ as $x\to x_0$ for the function $f(x)$ such that $f(x)/x \to 0$ as $x\to x_0$, i.e., $o(1) \to 0$ as $\alpha \to 0$, and $O(x)$ for the function $f(x)$ such that
$f(x)/x$ is bounded as $x \to 0$, i.e., $O(1)$ is a finite constant.

Along with Lorden's worst average detection delay $\WADD{\tau}$, defined in \eqref{LordenADD},
we can also consider the less pessimistic Pollak's performance measure \cite{PollakAS85}:
\[
    \SADD{\tau} := \sup_{\nu \geq 1}  \E{\nu}{\tau-\nu+1| \tau \geq \nu}.
\]
Pollak suggested the following minimax optimization problem in class $\mathcal{C}_\alpha$:
\begin{equation}\label{Pollak_problem}
\inf_{\tau\in \mathcal{C}_\alpha}  \SADD{\tau}.
\end{equation}

An alternative to CuSum is the Shiryaev-Roberts (SR) change detection procedure $\tau_{\text{SR}}$ based not on the maximization of the likelihood ratio over the unknown change point but on summation of likelihood ratios (i.e., on averaging over the uniform prior distribution).   As shown in \cite{tartakovsky_polpolunch_2012}, the SR procedure  is second-order asymptotically minimax with respect to Pollak's measure:
\[
\inf_{\tau\in \mathcal{C}_\alpha}  \SADD{\tau} = \SADD{\tau_{\text{SR}}} + O(1) \quad \text{as}~ \alpha\to 0.
\]
The CuSum procedure with a certain threshold $b_\alpha$ also has a second-order optimality property with respect to the risk $\SADD{\tau}$. A detailed numerical comparison of CuSum and SR procedures for i.i.d.\ models was performed in \cite{MoustPolTarCS09}.

\subsection{Information Bounds for Non-stationary Post-Change Observations}
\label{subsec:ext_ib}

In the case where both the pre- and post-change observations are independent and the post-change observations are non-stationary,  the log-likelihood ratio is:
\begin{equation}
\label{llr:def}
    Z_{n,k} = \log \frac{p_{1,n,k}(X_n)}{p_0(X_n)}
\end{equation}
where $n \geq k \geq 1$. Here $k$ is a hypothesized change-point and $X_n$ is drawn from the true distribution $\mathbb{P}_\nu$ ($\nu \in [1,\infty)$ or $\nu =\infty$).

In the classical i.i.d. model described in Section~\ref{subsec:stat_postc}, the cumulative KL-divergence after the change point increases linearly in the number of observations. We generalize this condition as follows. 
Let $\gfun_\nu: \R^+ \to \R^+$ be an increasing and continuous function, which we will refer to as \emph{growth function}. Note that the inverse of $\gfun_\nu$, denoted by $\gfun_\nu^{-1}$, exists and is also increasing and continuous. We assume that the expected sum of the log-likelihood ratios under $\mathbb{P}_\nu$, which corresponds to the cumulative KL-divergence for our non-stationary model,  matches the value of the growth function at all positive integers, i.e.,
\begin{equation}
\label{llr:growth}
    \gfun_\nu(n) = \sum_{i=\nu}^{\nu+n-1} \E{\nu}{Z_{i,\nu}},\forall n \geq 1
\end{equation}
Furthermore, we assume that $\E{\nu}{Z_{i,\nu}} > 0$ for all $i \geq \nu$
and that for each $x >0$
\begin{equation}
    \gfun^{-1}(x) := \sup_{\nu \geq 1} \gfun_\nu^{-1}(x)
\end{equation}
exists. Note that $\gfun^{-1}$ is also increasing and continuous. \vvv{We also assume that $\gfun_\nu(n)$ diverges for all $\nu \geq 1$, and thus 
$\gfun^{-1}(x)$ is properly defined on the entire positive real line.} In the special case where the post-change distribution is invariant to the change-point $\nu$, i.e., for $j\geq 0$, $p_{1,\nu+j,\nu}$ is not a function of $\nu$, we have $g \equiv g_\nu$ and $g^{-1} \equiv g_\nu^{-1}$ for all $\nu \geq 1$.

 \vvv{In order for change detection procedures to perform well it is necessary for the cumulative KL divergence between post- and pre-change distributions to grow sufficiently fast with $n$, e.g., faster than $\log n$. That is, the inverse $\gfun^{-1}(x)$ cannot grow too fast. This fact was 
discussed in \cite{TartakovskySISP98} for the hypothesis testing problem where it was shown that if $g(n)\sim \log n$, then the performance of Wald's SPRT is extremely poor. The same is true for change detection problems.} 

\vvv{The following key assumption on $\gfun^{-1}(x)$ guarantees the asymptotic optimality solution to the minimax problem in which we are interested:}
\begin{equation}
\label{llr:grow_cond}
    \log \gfun^{-1}(x) = o(x) \quad \text{as}~ x \to \infty.
\end{equation}

\vvv{To better understand this condition, we provide three special cases below:   
\begin{enumerate}
    \item The post-change observations are independent and stationary (as in the classical case). Here $\gfun^{-1}(x)$ is linear in $x$. Thus, $\log \gfun^{-1}(x) = O(\log x) = o(x)$, and condition~\eqref{llr:grow_cond} is always satisfied.
    \item The KL divergence  between the post- and pre-change distributions always increases (or increases asymptotically). Intuitively, this means that the post-change distributions increasingly drift away from that of the pre-change over time, as in Example~\ref{gex} below and in the pandemic monitoring example in Section~\ref{num-res:covid}. In this case, $\gfun^{-1}(x)$ grows at a slower than linear rate, and thus condition~\eqref{llr:grow_cond} is always satisfied.
    \item The KL divergence between the post- and pre-change distributions gradually decreases to 0. Intuitively, this means that the post-change distributions gradually recovers that of the pre-change over time. Condition~\eqref{llr:grow_cond} guarantees detection for slow enough recovery, specifically, when the post-change KL divergence satisfies $\E{\nu}{Z_{n,\nu}} \sim n^{-\theta}$ with $\theta < 1$ as the decay factor. In this case, $\log \gfun^{-1}(x) = O((1-\theta)^{-1} \log x) = o(x)$. Obviously, condition~\eqref{llr:grow_cond} fails if $\E{\nu}{Z_{n,\nu}} \sim n^{-1}$, i.e., when $g(n)\sim \log n$ and hence $\log \gfun^{-1}(x)\sim x$. We provide an example and some simulations for diminishing KL divergence between the post- and pre-change distributions in Section~\ref{num-res:pa_dim}.
\end{enumerate}
}

We should note that such a growth function $\gfun (n)$ has been adopted previously in sequential hypothesis testing with non-stationary observations \cite[Sec.~3.4]{tartakovsky_sequential}, \vvv{but not in QCD problem formulations such as the one considered here.}


The proof of asymptotic optimality is performed in two steps. First, we derive a first-order asymptotic (as $\alpha \to 0$) lower bound for the maximal expected detection delays $\inf_{\tau \in \Ccal}\WADD{\tau}$ and $\inf_{\tau \in \Ccal}\SADD{\tau}$. To this end, we need the following
right-tail condition for the log-likelihood ratio process:
\begin{equation}
\label{llr:upper}
    \sup_{\nu \geq 1} \Prob{\nu}{\max_{t \leq n} \sum_{i=\nu}^{\nu+t-1} Z_{i,\nu} \geq (1+\delta) \gfun_\nu(n)} \xrightarrow{n \to \infty} 0 \quad \forall \delta >0,
\end{equation}
assuming that for all $\nu \ge 1$
\[
\frac{\sum_{i=\nu}^{\nu+t-1} Z_{i,\nu}}{\gfun_\nu(t)} \xrightarrow[t\to\infty]{\text{in} ~ \mathbb{P}_\nu\text{-probability}} 1.
\]
At the second stage, we show that this lower bound is attained for the WL-CuSum procedure under the following left-tail condition
\begin{equation}
\label{llr:lower}
    \max_{t \geq \nu \geq 1} \Prob{\nu}{\sum_{i=t}^{t+n-1} Z_{i,t} \leq (1-\delta) \gfun_\nu(n) } \xrightarrow{n \to \infty} 0 \quad \forall \delta \in (0, 1).
\end{equation}

The following lemma provides sufficient conditions under which conditions \eqref{llr:upper} and \eqref{llr:lower} 
hold for the sequence of independent and non-stationary observations. Hereafter we use the notation $\Var{\nu}{Y}= \E{\nu}{Y^2 - \E{\nu}{Y}^2}$
for variance of the random variable $Y$ under distribution $\mathbb{P}_\nu$.

\begin{lemma}
\label{llr:lemma}
Consider the growth function $\gfun_\nu(n)$ defined in \eqref{llr:growth}. Suppose that the sum of variances of the log-likelihood ratios satisfies
\begin{equation}
\label{llr:var}
    \sup_{t \geq \nu \geq 1} \frac{1}{\gfun_\nu^2(n)} \sum_{i=t}^{t+n-1} \Var{\nu}{Z_{i,t}} \xrightarrow{n \to \infty} 0
\end{equation}
Then condition \eqref{llr:upper} holds. 

If, in addition, for all $\nu \geq 1$ and all positive integers $\Delta$,
\begin{equation}
\label{llr:tshift}
    \E{\nu}{Z_{i,\nu}} \leq \E{\nu}{Z_{i+\Delta,\nu+\Delta}},
\end{equation}
then condition \eqref{llr:lower} holds.
\end{lemma}

The proof is given in the appendix.

\begin{remark}
One can generalize condition \eqref{llr:tshift} in a way that either $\E{\nu}{Z_{i,\nu}} \leq \E{\nu}{Z_{i+\Delta,\nu+\Delta}}$ or
\begin{equation*}
    \frac{1}{\gfun_\nu(n)} \sum_{i=\nu}^{\nu+n-1} \left( \E{\nu}{Z_{i,\nu}} - \E{\nu}{Z_{i+\Delta,\nu+\Delta}}\right) = o(1)
\end{equation*}
holds for all positive integers $\Delta$.
\end{remark}

\begin{example}
\label{gex}
Consider the following Gaussian exponential mean-change (GEM) detection problem. Denote by ${\cal N}(\mu_0,\sigma_0^2)$ the Gaussian distribution with mean $\mu_0$ and variance $\sigma_0^2$. Let $X_1,\dots,X_{\nu-1}$ be distributed as ${\cal N}(\mu_0,\sigma_0^2)$, and for all $n \geq \nu$ let $X_n$ be distributed as ${\cal N}(\mu_0 e^{\theta(n-\nu)},\sigma_0^2)$. Here $\theta$ is some positive fixed parameter. The log-likelihood ratio is given by:
\begin{align}
\label{gex:llr}
    Z_{n,t} = \log \frac{p_{1,n,t}(X_n)}{p_0(X_n)} &= -\frac{(X_n-\mu_0 e^{\theta(n-t)})^2}{2 \sigma_0^2} + \frac{(X_n-\mu_0)^2}{2 \sigma_0^2} \nonumber\\
    &= \frac{\mu_0}{\sigma_0^2} (e^{\theta(n-t)} - 1) X_n - \frac{\mu_0^2 (e^{2 \theta (n-t)}-1)}{2 \sigma_0^2}.
\end{align}
Now, the growth function can be calculated as
\begin{equation}
\label{gex:growth}
    \gfun_\nu(n) = \sum_{i=\nu}^{\nu+n-1} \E{\nu}{Z_{i,\nu}} = \sum_{i=0}^{n-1} \frac{\mu_0^2}{2 \sigma_0^2} (e^{\theta i}-1)^2.
\end{equation}
Since the post-change distribution is invariant to the change-point $\nu$, $\gfun^{-1}(n) = \gfun_1^{-1}(n) = O(\log n) \implies \log \gfun^{-1}(n) = o(n)$, which satisfies \eqref{llr:grow_cond}. Also, the sum of variances of the log-likelihood ratios is
\begin{equation*}
    \sum_{i=t}^{t+n-1} \Var{\nu}{Z_{i,t}} = \sum_{i=t}^{t+n-1}  \frac{\mu_0^2}{\sigma_0^4} (e^{\theta(i-t)} - 1)^2 \Var{\nu}{X_i} = \frac{2}{\sigma_0^2} \gfun_\nu(n) = o(\gfun_\nu^2(n))
\end{equation*}
for all $t \geq \nu$, which establishes condition \eqref{llr:var}.
Further, for any $i \geq \nu$ and $\Delta \geq 1$,
\begin{align*}
    \E{\nu}{Z_{i+\Delta,\nu+\Delta}} &=  \frac{\mu_0}{\sigma_0^2} (e^{\theta(i-\nu)} - 1) \E{\nu}{X_{i+\Delta}} - \frac{\mu_0^2 (e^{2 \theta (i-\nu)}-1)}{2 \sigma_0^2}\\
    &\geq \frac{\mu_0}{\sigma_0^2} (e^{\theta(i-\nu)} - 1) \E{\nu}{X_{i}} - \frac{\mu_0^2 (e^{2 \theta (i-\nu)}-1)}{2 \sigma_0^2} = \E{\nu}{Z_{i,\nu}}
\end{align*}
which establishes condition \eqref{llr:tshift}.
\end{example}

The following theorem gives a lower bound on the worst-case average detection delays as $\alpha \to 0$ in class $\Ccal$.
\begin{theorem}
\label{llr:delay_lower_bound}
For $\delta\in(0,1)$ let 
\begin{equation}
\label{eq:thm1_h}
    h_\delta (\alpha) := \gfun^{-1} ((1-\delta) |\log\alpha|).
\end{equation} 
Suppose that $\gfun^{-1}(x)$ satisfies \eqref{llr:grow_cond}. Then for all $\delta\in(0,1)$ and some 
$\nu \ge 1$
\begin{equation}\label{supProb0}
\lim_{\alpha \to 0} \sup_{\tau\in\Ccal}  \Prob{\nu}{\nu \leq \tau < \nu + h_\delta (\alpha) } = 0
\end{equation}
and as $\alpha \to 0$,
\begin{equation} \label{LowerSADD} 
    \inf_{\tau \in \Ccal} \WADD{\tau} \geq \inf_{\tau \in \Ccal} \SADD{\tau} \geq \gfun^{-1}(\abs{\log{\alpha}}) (1+o(1)).
\end{equation}
\end{theorem}

\begin{proof}
Obviously, for any Markov time $\tau$,
\[
    \WADD{\tau} \geq \SADD{\tau} \ge  \E{\nu}{(\tau -\nu)^+} .
\]
Therefore, to prove the asymptotic lower bound \eqref{LowerSADD} we have to show that as $\alpha \to 0$,
\begin{equation}
\label{eq:thm1_main}
    \sup_{\nu \ge 1} \E{\nu}{(\tau -\nu)^+} \geq \gfun^{-1} (|\log{\alpha}|) (1+o(1)),
\end{equation}
where the $o(1)$ term on the right-hand side does not depend on $\tau$, i.e., uniform in $\tau \in \Ccal$.

To begin, let the stopping time $\tau \in  {\cal C}_\alpha$ and note that by Markov's inequality,  
\[
 \E{\nu}{(\tau -\nu)^+} \ge h_\delta (\alpha)  \Prob{\nu}{(\tau -\nu)^+ \ge h_\delta (\alpha)}.
\]
Hence, if assertion \eqref{supProb0} holds, then for some $\nu \ge 1$
\[
\inf_{\tau \in \Ccal} \Prob{\nu}{(\tau -\nu)^+ \ge h_\delta (\alpha)} = 1 -o(1) \quad \text{as}~ \alpha\to0.
\]
 This implies the asymptotic inequality
\begin{equation} \label{Exptauplus}
\inf_{\tau\in \Ccal} \E{\nu}{(\tau -\nu)^+}  \ge h_\delta (\alpha) (1+o(1)),
\end{equation}
which holds for an arbitrary $ \delta\in (0,1)$ and some $\nu$. Since by our assumption the function $ h_\delta (\alpha) $ is continuous, 
taking the limit $\delta \to 0$ and maximizing over $\nu\ge 1$ yields inequality \eqref{eq:thm1_main}.

It remains to prove \eqref{supProb0}.  
Changing the measure $\mathbb{P}_\infty \to \mathbb{P}_\nu$ and using Wald's likelihood ratio identity, we obtain the following chain of equalities and inequalities
for any $C >0$ and $\delta\in (0,1)$:
\begin{align*}
& \Prob{\infty}{\nu \leq \tau < \nu + h_\delta (\alpha)} = \E{\nu}{\ind{0 \le \tau -\nu < h_\delta(\alpha)} \exp\left(- \sum_{i=\nu}^\tau Z_{i,\nu} \right)}
\\
& \ge \E{\nu}{\ind{0 \le \tau -\nu < h_\delta(\alpha), \sum_{i=\nu}^\tau Z_{i,\nu} < C} \exp\left(- \sum_{i=\nu}^\tau Z_{i,\nu} \right) }
\\
&\ge e^{-C} \Prob{\nu}{0 \le \tau -\nu < h_\delta(\alpha), \max_{0 \le n-\nu <  h_\delta(\alpha) }\sum_{i=\nu}^n Z_{i,\nu} < C}
\\
&\ge e^{-C} \left(\Prob{\nu}{0 \le \tau -\nu < h_\delta(\alpha)} - \Prob{\nu}{\max_{0 \le n <  h_\delta(\alpha) }\sum_{i=\nu}^{\nu+n} Z_{i,\nu} \ge C}\right) ,
\end{align*} 
where the last inequality follows from the fact that $\Pr({\cal A} \cap {\cal B}) = \Pr({\cal A}) - \Pr({\cal B}^c)$ for any events ${\cal A}$ and ${\cal B}$, where 
${\cal B}^c$ is the complement event of ${\cal B}$. Setting $C=\gfun(h_\delta (\alpha)) (1+\delta) = (1-\delta^2) |\log\alpha|$ yields
\begin{equation}\label{Probnuh}
\Prob{\nu}{\nu \leq \tau < \nu + h_\delta (\alpha)} \le \kappa^{(\nu)}_{\delta,\alpha}(\tau) + \sup_{\nu \ge 1} \beta^{(\nu)}_{\delta,\alpha},
\end{equation}
where
\[
 \kappa^{(\nu)}_{\delta,\alpha}(\tau) =e^{(1-\delta^2) |\log\alpha|} \Prob{\infty}{0 \le \tau -\nu < h_\delta(\alpha)}
\]
and
\[
\beta^{(\nu)}_{\delta,\alpha}= \Prob{\nu}{\max_{0 \le n <  h_\delta(\alpha) }\sum_{i=\nu}^{\nu+n} Z_{i,\nu} \ge (1+\delta) \gfun(h_\delta (\alpha))}.
\]
Since $\gfun(h_\delta (\alpha))\to \infty$ as $\alpha \to 0$,  by condition \eqref{llr:upper}, 
\begin{equation}\label{subbetato0}
\sup_{\nu \ge 1} \beta^{(\nu)}_{\delta,\alpha}\to 0.
\end{equation}

Next we turn to the evaluation of the term $\kappa^{(\nu)}_{\delta,\alpha}(\tau)$ for any stopping time $\tau \in  {\cal C}_\alpha$. 
It follows from Lemma 2.1 in \cite[page 72]{tartakovsky_qcd2020} that for any $M < \alpha^{-1}$, there exists some $\ell \geq 1$ (possibly depending on $\alpha$) such that
\begin{equation}
\label{eq:thm1_lemma}
 \Prob{\infty}{\ell \leq \tau < \ell+M} \le \Prob{\infty}{\tau < \ell + M | \tau \geq \ell}  < M \, \alpha,
\end{equation}
so for some $\nu \ge 1$,
\[
    \kappa^{(\nu)}_{\delta,\alpha}(\tau) \le M  \alpha e^{(1-\delta^2) |\log\alpha|} = M\alpha^{\delta^2}.
\]
If we choose $M  \le M_\alpha= \floor{h_\delta(\alpha)^2}\Big|_{\delta = 0} = \floor{(\gfun^{-1}(|\log\alpha|))^2}$, then for all sufficiently small $\alpha$,
\begin{equation*}
    \log M \leq 2 \log \gfun^{-1}(|\log\alpha|) = o(|\log\alpha|)
\end{equation*}
so that the condition \eqref{llr:grow_cond} is satisfied. Furthermore, 
\begin{equation*}
    M_\alpha \, \alpha^p \xrightarrow{\alpha \to 0} 0 \quad \text{as}~ \alpha\to 0
\end{equation*}
for any $p > 0$. To see this, assume for purpose of contradiction that there exists some $p_0 > 0$ and $c_0 > 0$ such that 
$\lim_{\alpha \to 0} M_\alpha \alpha^{p_0} = c_0$. Then, since $\lim_{\alpha \to 0} \alpha^{-p_0} \neq 0$, 
$\lim_{\alpha \to 0} \log M_\alpha = p_0 \lim_{\alpha \to 0} |\log \alpha| + \log c_0$ and thus $\log M_\alpha \neq o(|\log\alpha|)$. Hence, it follows that for some 
$\nu \ge 1$, which may depend on $\alpha$, as $\alpha \to 0$
\begin{equation}\label{kappainfto0}
\inf_{\tau\in\Ccal} \kappa^{(\nu)}_{\delta,\alpha}(\tau) \le M_\alpha \alpha^{\delta^2} \to 0.
\end{equation}

Combining \eqref{Probnuh}, \eqref{subbetato0}, and \eqref{kappainfto0} we obtain that for some $\nu\ge 1$
\[
\Prob{\nu}{\nu \leq \tau < \nu + h_\delta (\alpha)} \le M_\alpha \alpha^{\delta^2}  + \sup_{\nu \ge 1} \beta^{(\nu)}_{\delta,\alpha} = o(1),
\]
where the $o(1)$ term is uniform over all $\nu\ge 1$. This yields assertion \eqref{supProb0}, and the proof is complete.
\end{proof}

\subsection{Asymptotically Optimal Detection for Non-stationary Post-Change Observations with Known Distributions}
\label{subsec:asym-opt-det}

Recall that under the classical setting, Page's CuSum procedure (in \eqref{defstoppingrule}) is optimal and has the following structure:
\begin{equation}
    \tau_{\text{Page}}\left(b\right) = \inf \left\{n: \max_{1\leq k \leq n+1} \sum_{i=k}^n Z_i \geq b \right\}
\end{equation}
where $Z_i$ is the log-likelihood ratio when the post-change distributions are stationary. When the post-change distributions are potentially non-stationary, the CuSum stopping rule is defined similarly as:
\begin{equation}
\label{TC:def}
    \tau_C\left(b\right) := \inf \left\{n:\max_{1 \leq k \leq n+1} \sum_{i=k}^{n} Z_{i,k} \geq b \right\}
\end{equation}
where $Z_{i,k}$ represents the log-likelihood ratio between densities $p_{1,i,k}$ and $p_0$ for observation $X_i$ (defined in \eqref{llr:def}). Here $i$ is the time index and $k$ is the hypothesized change point. Note that if the post-change distributions are indeed stationary, i.e., $p_{1,i,k} \equiv p_1$, we would get $Z_{i,k} \equiv Z_i$ for all $k \leq i$, and thus $\tau_C \equiv \tau_{\text{Page}}$.

Page's classical CuSum algorithm admits a recursive way to compute its test statistic. Unfortunately, despite having independent observations, the test statistic in \eqref{TC:def} cannot be computed recursively, even for the special case where the post-change distribution is invariant to the change-point as in Example~\ref{gex}.

\begin{example}
Consider the Gaussian Exponential Mean-Change problem defined in Example~\ref{gex}. Suppose $\mu_0 = \sigma_0^2 = \theta = 1$. Then, the log-likelihood ratio is given by
\[ Z_{n,t} = (e^{n-t} - 1) X_n - \frac{e^{2 (n-t)} - 1}{2}. \]
Note that $Z_{n,t}$ is a (linear) function of $X_n$. Consider the following realization:
\[ X_1 = 1, \quad X_2 = 0, \quad X_3 = 10. \]
It can be verified that
\[ \arg \max_{1 \leq k \leq 3} \sum_{i=k}^{2} Z_{i,k} = 2,~\text{and}~\arg \max_{1 \leq k \leq 4} \sum_{i=k}^{3} Z_{i,k} = 1. \]
Note that maximizer $k^*$ goes backward in time in this case, in contrast to what happens when both the pre- and post-change observations follow i.i.d. models. The test statistic at time $n=2$ is a function of only $X_2$, and this is insufficient to construct the test statistic at time $n=3$, which is a function of $X_1$, in addition to being a function of $X_2$, and $X_3$.
\end{example}

For computational tractability we therefore consider a window limited version of the CuSum procedure in \eqref{TC:def}:
\begin{equation}
\label{TC:test}
    \td{\tau}_C\left(b\right) := \inf \left\{n:\max_{n-m \leq k \leq {n+1}} \sum_{i=k}^n Z_{i,k} =: W(n) \geq b \right\}
\end{equation}
where $m$ is the window size. For $n < m$ maximization is performed over $1\le k \le n$. In the asymptotic setting, $m=m_\alpha$ depends on $\alpha$ and should go to infinity as $\alpha \to 0$ with certain appropriate rate. 
Specifically, following a similar condition that Lai~\cite{lai1998} used in the asymptotically stationary case, we shall require that $m_\alpha \to \infty$ as $\alpha \to 0$ 
in such a way that
\begin{equation}
\label{TC:m_alpha}
    \liminf_{\alpha \to 0} m_\alpha / \gfun^{-1}(\abs{\log\alpha}) > 1.
\end{equation}

Since the range for the maximum is smaller in $\td{\tau}_C(b)$ than in $\tau_C(b)$, given any realization of $X_1,X_2,\ldots$, if the test statistic of $\td{\tau}_C(b)$ crosses the threshold $b$ at some time $n$, so does that of $\tau_C(b)$. Therefore, for any fixed threshold $b > 0$,
\begin{equation} \label{eq:tautautil}
    \tau_C(b) \leq \td{\tau}_C(b)
\end{equation}
almost surely.

In the following, we first control the asymptotic false alarm rate of $\td{\tau}_C(b)$ with an appropriately chosen threshold in Lemma~\ref{TC:fa}. Then we obtain asymptotic approximation of the expected detection delays of $\td{\tau}_C(b)$ in Theorem \ref{TC:delay}. Finally, we combine these two results and provide an asymptotically optimal solution to the problem in \eqref{prob_def} in Theorem~\ref{TC:asymp_opt}.



\begin{lemma}
\label{TC:fa}
Suppose that $b_\alpha = \abs{\log\alpha}$. Then
\begin{equation}
\label{FARCUSUMm}
 \FAR{\td{\tau}_C(b_\alpha)} \leq \alpha \quad \text{for all} ~ \alpha \in (0,1),
\end{equation}
i.e., $\td{\tau}_C(b_\alpha) \in \Ccal$.
\end{lemma}

\begin{proof}
Define the statistic 
\[
    R_n = \sum_{k=1}^n \exp\left(\sum_{i=k}^n Z_{i,k}\right), \quad R_0=0
\]
and the corresponding stopping time $T_b:=\inf\{n: R_n \ge e^b\}$. We now show that 
$\E{\infty}{T_b} \ge e^b$, which implies that $\E{\infty}{\td{\tau}_C(b)} \ge e^b$ for any $b>0$ since, evidently, 
$\td{\tau}_C(b) \ge T_b$ for any $b>0$. Recall that ${\cal F}_n=\sigma(X_\ell, 1\le \ell \le n)$ denotes a sigma-algebra generated by $(X_1,\dots,X_n)$. 
Since $\E{\infty}{e^{Z_{n,k}}|{\cal F}_{n-1}}=1$, it is easy to see that
\[
    \E{\infty}{R_n | {\cal F}_{n-1}} = 1 + R_{n-1} \quad \text{for}~ n \ge 1.
\]
Consequently, the statistic $\{R_n-n\}_{n \ge 1}$ is a zero-mean $(\mathbb{P}_\infty,{\cal F}_n)$-martingale. It suffices to assume that $\E{\infty}{T_b}<\infty$ 
since otherwise the statement is trivial. Then, $\E{\infty}{R_{T_b}-T_b}$ exists and also
\[
    \liminf_{n\to\infty} \int_{\{T_b >n\}} |R_n - n| \mathrm{d} \mathbb{P}_\infty = 0
\]
since $0 \le R_n < e^b$ on the event $\{T_b >n\}$. Hence, we can apply the optional sampling theorem 
(see, e.g. \cite[Th 2.3.1, page 31]{tartakovsky_sequential}), which yields 
$\E{\infty}{R_{T_b}} = \E{\infty}{T_b}$. Since $R_{T_b} \ge e^b$ it follows that $\E{\infty}{\td{\tau}_C(b)} \ge \E{\infty}{T_b} \ge e^b$.

Now, setting $b_\alpha = \abs{\log\alpha}$ implies the inequality   
\begin{equation} \label{FARCUSUMbalpha}
\E{\infty}{\td{\tau}_C(b_\alpha)} \ge e^{b_\alpha} =  \frac{1}{\alpha}
\end{equation}
(for any $m_\alpha \ge 1$), and therefore \eqref{FARCUSUMm} follows.
\end{proof}

The following result establishes asymptotic performance of the WL-CuSum procedure given in \eqref{TC:test} for large threshold values.

\begin{theorem}
\label{TC:delay}
Fix $\delta \in (0,1)$ and let $N_{b,\delta} := \lfloor g^{-1}(b /(1-\delta)) \rfloor$.
Suppose that in the WL-CuSum procedure the size of the window $m=m_b$ diverges (as $b \to \infty$) in such a way that
\begin{equation}\label{Condmb}
    m_b \ge N_{b,\delta} (1+o(1)).
\end{equation}
Further, suppose that conditions \eqref{llr:upper} and \eqref{llr:lower}  hold for $Z_{n,k}$ when $n \geq k \geq 1$. Then, as $b \to \infty$,
\begin{equation}\label{SADDasympt}
    \SADD{\td{\tau}_C(b)}  \sim  \WADD{\td{\tau}_C(b)} \sim \gfun^{-1} (b) .
\end{equation}
\end{theorem}

\begin{proof}
Since $\FAR{\td{\tau}_C(b)} \le e^{-b}$, the WL-CuSum procedure $\td{\tau}_C(b)$ belongs to class $\Ccal$ with $\alpha=e^{-b}$. 
Hence, replacing $\alpha$ by $e^{-b}$ in  the asymptotic lower bound \eqref{LowerSADD} 
in Theorem~\ref{llr:delay_lower_bound}, we obtain that under condition \eqref{llr:upper} the following asymptotic lower bound holds:
\begin{equation}\label{LBtaumb}
 \liminf_{b\to\infty} \frac{\WADD{\td{\tau}_C(b))}}{g^{-1}(b)}  \ge \liminf_{b\to\infty} \frac{\SADD{\td{\tau}_C(b))}}{g^{-1}(b)} \ge  1 .
\end{equation}
Thus, to establish \eqref{SADDasympt} it suffices to show that under condition \eqref{llr:lower} as $b\to\infty$
\begin{equation}\label{UpperSADD}
\WADD{\td{\tau}_C(b))} \leq g^{-1}(b) (1+o(1)).
\end{equation}

Note that we have the following chain of equalities and inequalities:
\begin{align} \label{ExpkTAplus}
&\E{\nu}{(\td{\tau}_C(b))-\nu)^+ | {\cal F}_{\nu-1}}    \nonumber
\\
&= \sum_{\ell=0}^{\infty} \int_{\ell N_{b,\delta}}^{(\ell+1) N_{b,\delta}}  \Prob{\nu}{\td{\tau}_C(b)-\nu > t | {\cal F}_{\nu-1}} \, \D t    \nonumber\\
& \le N_{b,\delta} + \sum_{\ell =1}^{\infty} \int_{\ell N_{b,\delta}}^{(\ell +1) N_{b,\delta}} \Prob{\nu}{\td{\tau}_C(b)-\nu > t |{\cal F}_{\nu-1}} \, \D t  \nonumber\\
& \le  N_{b,\delta} + \sum_{\ell=1}^{\infty} \int_{\ell N_{b,\delta}}^{(\ell+1) N_{b,\delta}}   \Prob{\nu}{\td{\tau}_C(b)-\nu > \ell N_{b,\delta} | {\cal F}_{\nu-1}} \, \D t  \nonumber\\
& = N_{b,\delta} \left(1 + \sum_{\ell=1}^{\infty}  \Prob{\nu}{\td{\tau}_C(b)-\nu > \ell N_{b,\delta} | {\cal F}_{\nu-1}}\right) .  
\end{align}
Define $\lambda_{n,k} := \sum_{i=k}^n Z_{i,k}$ and $K_{n} := \nu+n N_{b,\delta}$. We have $W(n) = \max_{n-m_b < k \le n} \lambda_{k,n}$. 
Since by condition \eqref{Condmb} $m_b >  N_{b,\delta}$ (for a sufficiently large $b$), for any $n \ge 1$,
\[
    W(\nu+n N_{b,\delta}) \ge \lambda_{K_{n}, K_{n-1}}
\]
and we have
\begin{align}\label{Needit}
&\Prob{\nu}{\td{\tau}_C(b)-\nu > \ell N_{b,\delta} | {\cal F}_{\nu-1} } \nonumber
\\
& =\Prob{\nu}{W(1) <  b, \dots, W(\nu+\ell N_{b,\delta}) <b | {\cal F}_{\nu-1}} \nonumber
\\
& \le  \Prob{\nu}{W(\nu + N_{b,\delta}) < b, \dots, W(\nu+\ell N_{b,\delta}) <b  | {\cal F}_{\nu-1}} \nonumber
\\
& \le  \Prob{\nu}{\lambda_{K_1, K_0}< b, \dots, \lambda_{K_{\ell}, K_{\ell-1}} < b  | {\cal F}_{\nu-1}}\nonumber
\\
&= \prod_{n=1}^\ell \Prob{\nu}{\lambda_{K_{n}, K_{n-1}}< b} ,
\end{align}
where the last equality follows from independence of the increments of $\{\lambda_{t,n}\}_{n \ge t}$.

By condition \eqref{llr:lower}, for a sufficiently large $b$ there exists a small $\eps_b$ such that
\[
    \Prob{\nu}{\lambda_{K_{n}, K_{n-1}} < b} \le \eps_b, \quad \forall n \ge 1.
\]
Therefore, for any $\ell \ge 1$, 
\[
    \Prob{\nu}{\td{\tau}_C(b)-\nu > \ell N_{b_\alpha,\delta} | {\cal F}_{\nu-1}}  \le \varepsilon_b^\ell.
\]
Combining this inequality with \eqref{ExpkTAplus} and using the fact that $\sum_{\ell=1}^\infty \eps_b^\ell = \eps_b(1-\eps_b)^{-1}$ , 
we obtain
\begin{align}\label{Momineqrho}
    \E{\nu}{(\td{\tau}_C(b)-\nu)^+ | {\cal F}_{\nu-1} } \le N_{b,\delta} \left(1 + \frac{\eps_b}{1-\eps_b}\right) = \frac{\lfloor g^{-1}(b /(1-\delta)) \rfloor}{1-\varepsilon_b}.
\end{align}
Since the right-hand side of this inequality does not depend on $\nu$, $g^{-1}(b /(1-\delta))\to \infty$ as $b\to \infty$ and $\eps_b$ and $\delta$ can be arbitrarily small numbers, this implies the upper bound \eqref{UpperSADD}. The proof is complete. 
\end{proof}

Using Lemma~\ref{TC:fa} and Theorem~\ref{TC:delay}, we obtain the following asymptotic result which establishes asymptotic optimality of the WL-CuSum procedure and its asymptotic operating characteristics.

\begin{theorem}
\label{TC:asymp_opt}
Suppose that threshold $b_\alpha$ is so selected that $b_\alpha \sim|\log \alpha|$ as $\alpha\to 0$, in particular as $b_\alpha =|\log \alpha|$. 
Further, suppose that left-tail \eqref{llr:upper} and right-tail \eqref{llr:lower} conditions hold for $Z_{n,k}$ when $n \geq k \geq 1$.  
Then, the WL-CuSum procedure in \eqref{TC:test} with the window size $m_\alpha$ that satisfies the condition 
\begin{equation}\label{Condmalpha}
m_\alpha \ge g^{-1}(|\log\alpha|) (1+o(1)) \quad \text{as}~ \alpha \to 0
\end{equation}
solves the problems \eqref{prob_def} and \eqref{Pollak_problem} asymptotically to first order as $\alpha \to 0$, i.e., 
\begin{equation} \label{FOAOCUSUM}
\begin{split}
\inf_{\tau\in \Ccal} \WADD{\tau} & \sim  \WADD{\td{\tau}_C(b_\alpha)} ,
\\
\inf_{\tau\in \Ccal} \SADD{\tau} & \sim \SADD{\td{\tau}_C(b_\alpha)} 
\end{split}
\end{equation}
and
\begin{equation} \label{FOAPPRCUSUM} 
   \SADD{\td{\tau}_C(b_\alpha)} \sim  \WADD{\td{\tau}_C(b_\alpha)}  \sim \gfun^{-1} (\abs{\log{\alpha}}).
\end{equation}
\end{theorem}

\begin{proof}
Let $b_\alpha$ be so selected that  $\FAR{\td{\tau}_C(b_\alpha)}\leq \alpha$ and $b_\alpha \sim |\log \alpha|$ as $\alpha\to0$. Then by 
Theorem~\ref{TC:delay}, as $\alpha\to0$
\[
 \SADD{\td{\tau}_C(b_\alpha)}  \sim  \WADD{\td{\tau}_C(b_\alpha)} \sim \gfun^{-1} (|\log \alpha|) .
\]
Comparing these asymptotic equalities with the asymptotic lower bound \eqref{LowerSADD} in
Theorem~\ref{llr:delay_lower_bound} immediately yields asymptotics \eqref{FOAOCUSUM} and \eqref{FOAPPRCUSUM}. 
In particular, if $b_\alpha = \abs{\log\alpha}$, then by Lemma~\ref{TC:fa}  $\FAR{\td{\tau}_C(b_\alpha)}\leq \alpha$, and therefore the assertions hold.\end{proof}

\begin{remark}
Clearly, the asymptotic optimality result still holds in the case where no window is applied, i.e., $m_\alpha = n-1$.
\end{remark}

\begin{example}
Consider the same setting as in Example~\ref{gex}. We have shown that conditions \eqref{llr:var} and \eqref{llr:tshift} hold in this setting, and thus \eqref{llr:upper} and \eqref{llr:lower} also hold by Lemma~\ref{llr:lemma}. Considering the growth function $\gfun_\nu (n)$ given in \eqref{gex:growth}, as $n \to \infty$, we obtain
\begin{equation*}
    \gfun_\nu (n) = \sum_{i=0}^{n-1} \frac{\mu_0^2}{2 \sigma_0^2} (e^{\theta i}-1)^2 = \frac{\mu_0^2}{2 \sigma_0^2} e^{2 \theta (n-1)} (1+o(1)).
\end{equation*}
Thus, as $y \to \infty$,
\begin{equation*}
    \gfun^{-1}(y) = \frac{1}{2 \theta} \log\left( \frac{2 \sigma_0^2}{\mu_0^2} y\right) (1+o(1))
\end{equation*}
and if $b_\alpha = |\log \alpha|$ or more generally $b_\alpha\sim |\log \alpha|$ as $\alpha\to 0$ we obtain
\begin{align}
\label{gex:perf}
    \WADD{\td{\tau}_C(b_\alpha)} &= \frac{1}{2 \theta} \log\left( \frac{2 \sigma_0^2 }{\mu_0^2} \abs{\log \alpha}\right) (1+o(1)) \nonumber\\
    &= O\left(\frac{1}{2 \theta} \log(\abs{\log \alpha})\right).
\end{align}
\end{example}

\section{Asymptotically Optimum Procedure for Non-Stationary Post-Change Observations with Parametric Uncertainty}
\label{sec:unknown-param}

We now study the case where the evolution of the post-change distribution is parametrized by an unknown but deterministic parameter $\theta \in \R^d$. Let $X_\nu, X_{\nu+1}, \dots$
each have density $p_{1,0}^{\theta},p_{1,1}^{\theta},\dots$, respectively, with respect to the common non-degenerate measure $\mu$, when post-change parameter is $\theta$. Let $\mathbb{P}_{k,\theta}$ and $\mathbb{E}_{k,\theta}$ denote, respectively, the probability measure on the entire sequence of observations and expectation when the change point is $\nu=k<\infty$ and the post-change parameter is $\theta$. Let $\Theta \subset \R^d$ be an open and bounded set of parameter values. For any $n \geq k$ and $\theta \in \Theta$ the log-likelihood ratio process is given by
\begin{equation}
\label{TG:llr}
    Z_{n,k}^{\theta} = \log \frac{p_{1,n,k}^{\theta}(X_n)}{p_0(X_n)} .
\end{equation}
Also, the growth function in \eqref{llr:growth} is redefined as
\begin{equation}
\label{TG:growth}
    \gfun_{\nu,\theta}(n) = \sum_{i=\nu}^{\nu+n-1} \E{\nu,\theta}{Z^{\theta}_{i,\nu}},\forall n \geq 1
\end{equation}
and it is assumed that $\gfun_\theta^{-1}(x) = \sup_{\nu \geq 1} \gfun_{\nu,\theta}^{-1}(x)$ exists. It is also assumed that
\begin{equation}
\label{TG:growth_cond}
    \log \gfun_\theta^{-1}(x) = o(x) \quad \text{as}~ x \to \infty.
\end{equation}

The goal in this section is to solve the optimization problems \eqref{prob_def} and \eqref{Pollak_problem} asymptotically as $\alpha \to 0$ under parameter uncertainty. More specifically, for $\theta \in \Theta$, define Lorden's and Pollak's worst-case expected detection delay measures
\[
    \WADDth{\tau}  := \esssup \sup_{\nu \geq 1}  \E{\nu,\theta}{(\tau-\nu+1)^+ | {\cal F}_{\nu-1}}
\]
and 
\[
    \SADDth{\tau} := \sup_{\nu \geq 1}  \E{\nu, \theta}{\tau-\nu+1| \tau \geq \nu}
\]
and the corresponding asymptotic optimization problems: find a change detection procedure $\tau^*$ that minimizes these measures to first order in class 
$\Ccal$, i.e., for all $\theta\in\Theta$,
\begin{equation}\label{AsymptProblems}
    \lim _{\alpha \to 0}\frac{ \inf_{\tau\in \Ccal} \WADDth{\tau}}{\WADDth{\tau^*}} = 1, \quad 
    \lim_{\alpha\to0}\frac{\inf_{\tau\in \Ccal}  \SADDth{\tau}}{\SADDth{\tau^*}} = 1.
\end{equation}

Consider the following WL-GLR-CuSum change detection procedure
\begin{equation}
\label{TG:test}
    \td{\tau}_G\left(b\right) := \inf \left\{n:\max_{n-m_b \leq k \leq {n+1}} \sup_{\theta \in \Theta_b} \sum_{i=k}^n Z_{i,k}^{\theta} \geq b \right\} ,
\end{equation}
where $\Theta_b \nearrow \Theta$ as $b \nearrow \infty$. For $n < m_b$ maximization is performed over $1\le k \le n$. Therefore, it is guaranteed that $\theta \in \Theta_b$ for all large enough $b$. Since we are interested in class 
$\Ccal=\{\tau: \FAR{\tau}\le \alpha\}$, in which case both threshold $b=b_\alpha$ and window size $m_b=m_\alpha$ are the functions of $\alpha$, we will write
$\Theta_{b}=\Theta_\alpha$ and suppose that $\Theta_\alpha \subset \R^d$ is compact for each $\alpha$. Hereafter we omit the dependency of 
$\hat{\theta}_{n,k}$ on $\alpha$ for brevity. In this paper, we focus on the case where $\Theta_\alpha$ is continuous for all $\alpha$'s. 
The discrete case is simpler and will be considered elsewhere.

The following assumption is made to guarantee the existence of an upper bound on FAR.
\begin{assumption}
\label{TG:smooth}
There exists $\eps > 0$ such that for any large enough $b > 0$,
\begin{equation}
\label{TG:llr_sec_der}
    \Prob{\infty}{\max_{(k,n):k \leq n \leq k+m_b} \sup_{\theta: \norm{\theta-\hat{\theta}_{n,k}} < b^{-\frac{\eps}{2}}} \emax{- \nabla_\theta^2 \sum_{i=k}^n Z_{i,k}^{\theta}} \leq 2 b^{\eps}} \geq 1 - \xi_b
\end{equation}
where $\emax{A}$ represents the maximum absolute eigenvalue of a symmetric matrix $A$ and $\xi_b \searrow 0$ as $b \nearrow \infty$.
\end{assumption}


\begin{example}
Consider again the Gaussian exponential mean-change detection problem in Example~\ref{gex}. Now we consider the case where the exact value of the post-change exponent coefficient $\theta$ is unknown and belongs to $\Theta = [\Theta_{\text{min}},\Theta_{\text{max}}]$. Note that $\theta$ characterizes the entire post-change evolution rather than a single post-change distribution.
We shall verify Assumption~\ref{TG:smooth} below.

Recalling the definition of log-likelihood ratio given in \eqref{gex:llr}, for any $\theta \in \Theta$ and $k \leq i \leq n$ where $n-k\leq m_b$, we have 
\begin{align}
    -\frac{\partial^2}{\partial \theta^2} Z^\theta_{i,k} &= -\frac{\partial^2}{\partial \theta^2} \left(\frac{\mu_0}{\sigma_0^2} (e^{\theta(i-k)} - 1) X_i - \frac{\mu_0^2 (e^{2 \theta (i-k)}+1)}{2 \sigma_0^2}\right) \nonumber\\
    &= -\frac{\mu_0}{\sigma_0^2} (i-k)^2 e^{\theta(i-k)} X_i + 2 (i-k)^2 \frac{\mu_0^2 e^{2 \theta (i-k)}}{\sigma_0^2}\nonumber\\
    &= \frac{\mu_0}{\sigma_0^2} (i-k)^2 e^{\theta(i-k)} (2 \mu_0 e^{\theta(i-k)} - X_i).
\end{align}
Therefore,
\begin{align}
\label{gex:sec_der}
    &\max_{(k,n):k \leq n \leq k+m_b} \sup_{\theta \in \Theta} \abs{-\frac{\partial^2}{\partial \theta^2} \sum_{i=k}^n Z^\theta_{i,k}}\nonumber\\
    &= \sup_{\theta \in \Theta} \max_{(k,n):k \leq n \leq k+m_b} \frac{\mu_0}{\sigma_0^2} \abs{\sum_{i=k}^n (i-k)^2 e^{\theta(i-k)} (2 \mu_0 e^{\theta (i-k)} - X_i)} \nonumber\\
    &\leq \sup_{\theta \in \Theta} \frac{\mu_0}{\sigma_0^2} m_b^2 e^{\theta m_b} \left(2 \mu_0 m_b e^{\theta m_b} + \max_{(k,n):k \leq n \leq k+m_b} \abs{\sum_{i=k}^n X_i}\right) \nonumber\\
    &\stackrel{(*)}{\leq} \sup_{\theta \in \Theta} \frac{4\mu_0^2}{\sigma_0^2} m_b^3 e^{2 \theta m_b} \leq \frac{4\mu_0^2}{\sigma_0^2} m_b^3 e^{2 \Theta_{\text{max}} m_b}
\end{align}
where $(*)$ is true provided that
\[
    \max_{(k,n):k \leq n \leq k+m_b} \abs{\sum_{i=k}^n X_i} < 2 \mu_0 m_b e^{\theta m_b}.
\]
Since $X_i$'s are i.i.d. under $\mathbb{P}_\infty$, $\sum_{i=k}^n X_i$ has a Gaussian distribution with mean $\leq (m_b+1) \mu_0$ and variance $\leq (m_b+1) \sigma_0^2$.
Therefore, for any $\theta \in \Theta$,
\begin{align*}
   & \Prob{\infty}{\max_{(k,n):k \leq n \leq k+m_b} \abs{\sum_{i=k}^n X_i} > 2 \mu_0 m_b e^{\theta m_b}}\\
   &\leq \Prob{\infty}{\abs{\sum_{i=1}^{m_b} X_i} > 2 \mu_0 m_b e^{\theta m_b}} \\
   &= 2 Q\left( \frac{2 \mu_0 m_b e^{\theta m_b} - m_b \mu_0}{\sigma_0\sqrt{m_b+1}} \right)\\
   &\leq 2\exp\left(-\frac{2 \mu_0^2 m_b^2 (e^{\theta m_b}-1)^2}{\sigma_0^2 (m_b+1)}\right) \searrow 0~\text{as $b \to \infty$}
\end{align*}
where $Q(x)= (2\pi)^{-1/2} \int_{x}^\infty e^{-t^2/2} \D t$ is the standard Q-function.

Recalling the condition in \eqref{Condmb} on the window size and using the formula  \eqref{gex:perf} for the worst-case expected delay, we obtain that if we set
\begin{equation*}
    m_b = \frac{1}{2 \Theta_{\text{min}}}\log b
\end{equation*}
then
\begin{equation*}
    \frac{4\mu_0^2}{\sigma_0^2} m_b^3 e^{2 \Theta_{\text{max}} m_b} \sim (\log b)^3 b^{\Theta_{\text{max}}/\Theta_{\text{min}}}.
\end{equation*}
Then Assumption~\ref{TG:smooth} holds when $\eps = (1+\delta)\Theta_{\text{max}}/\Theta_{\text{min}}$ with arbitrary $\delta>0$.
\end{example}

Note that  $\WADDth{\td{\tau}_G(b)} \leq \WADDth{\td{\tau}_C(b)}$ for any threshold $b > 0$. In order to establish asymptotic optimality of the WL-GLR-CuSum procedure we need the following lemma that allows us to select threshold $b=b_\alpha$ in such a way that the FAR of $\td{\tau}_G(b)$ is controlled at least asymptotically.

\begin{lemma}
\label{TG:fa}
Suppose that the log-likelihood ratio $\{Z_{n,k}^{\theta}\}_{n\ge k}$ satisfies \eqref{TG:llr_sec_der}. Then, as $b\to\infty$, 
\begin{equation}\label{FARGLR}
    \FAR{\td{\tau}_G(b)} \le |\Theta_\alpha| C_d^{-1} b^{\frac{\eps d}{2}} e^{1-b} (1+o(1)),
\end{equation}
where $C_d =\frac{\pi^{d/2}}{\Gamma(1+d/2)}$ is a constant that does not depend on $\alpha$.
Consequently, if $b=b_\alpha$ satisfies equation
\begin{equation}
\label{TG:thr}
    |\Theta_\alpha| C_d^{-1} b_\alpha^{\frac{\eps d}{2}} e^{1-b_\alpha} = \alpha,
\end{equation}
then $\FAR{\td{\tau}_G(b_\alpha)} \le \alpha (1+o(1))$ as $\alpha\to0$. 
\end{lemma}

\begin{remark}
Since $\abs{\Theta_\alpha}  \leq \abs{\Theta}  < \infty$, it follows from \eqref{TG:thr} that $b_\alpha \sim \abs{\log\alpha}$ as $\alpha \to 0$. 
\end{remark}

The proof of Lemma~\ref{TG:fa} is given in the appendix. The following theorem establishes asymptotic optimality properties of the WL-GLR-CuSum detection procedure.


\begin{theorem}
\label{TG:asymp_opt}
Suppose that threshold $b = b_\alpha$ is so selected that $\FAR{\td{\tau}_C(b_\alpha)}\leq \alpha$ or at least so that 
$\FAR{\td{\tau}_C(b_\alpha)}\leq \alpha (1+o(1))$ and $b_\alpha \sim |\log \alpha|$ as $\alpha \to 0$, in particular from equation \eqref{TG:thr} in Lemma~\ref{TG:fa}.
Further, suppose that conditions \eqref{llr:upper}, \eqref{llr:lower} and \eqref{TG:llr_sec_der} hold for $\{Z_{n,k}\}_{n \ge k}$.  
Then, the WL-GLR-CuSum procedure $\td{\tau}_G(b_\alpha)$ defined by \eqref{TG:test} with the window size $m_\alpha$ that satisfies the condition \eqref{Condmalpha} solves first-order asymptotic optimization problems \eqref{AsymptProblems} uniformly for all parameter values
$\theta \in \Theta$, and
\begin{equation} \label{FOAPPRGCUSUM}
   \SADDth{\td{\tau}_G(b_\alpha)} \sim  \WADDth{\td{\tau}_G(b_\alpha)}  \sim \gfun_\theta^{-1} (\abs{\log{\alpha}}),\quad \forall \theta \in \Theta.
\end{equation}
as $\alpha \to 0$.
\end{theorem}

\begin{proof}
Evidently, for any $\theta \in \Theta$ and any threshold $b>0$,
\begin{equation*}
    \WADDth{\td{\tau}_G(b)} \leq \WADDth{\td{\tau}_C(b)}, \quad  \SADDth{\td{\tau}_G(b)} \leq \SADDth{\td{\tau}_C(b)}.
\end{equation*}
Let $b=b_\alpha$ be so selected that $\FAR{\td{\tau}_G(b_\alpha)} \leq \alpha$ and $b_\alpha \sim |\log \alpha|$ as $\alpha \to 0$. Then it follows from the asymptotic approximations \eqref{FOAPPRCUSUM} in Theorem~\ref{TC:asymp_opt} that, as $\alpha \to 0$, 
\[
 \SADDth{\td{\tau}_G(b_\alpha)} \le \WADDth{\td{\tau}_G(b_\alpha)} \le \gfun_\theta^{-1} (|\log \alpha|) (1+o(1)) .
\]
Comparing these asymptotic inequalities with the asymptotic lower bound \eqref{LowerSADD} in
Theorem~\ref{llr:delay_lower_bound}, immediately yields \eqref{FOAPPRGCUSUM}, which is asymptotically the best one can do to first order according to
Theorem~\ref{llr:delay_lower_bound}. 

In particular, if $b_\alpha$ is found from equation \eqref{TG:thr}, then $b_\alpha \sim |\log \alpha|$ as $\alpha\to0$ and by 
Lemma~\ref{TG:fa} $\FAR{\td{\tau}_G(b_\alpha)}\leq \alpha (1+o(1))$, and therefore the assertions hold. 
\end{proof}

\section{Extensions to Pointwise Optimality and Dependent Non-homogeneous Models}
\label{sec:ext}

The measure of FAR that we have used in this paper (see \eqref{eq:FAR_def}) is the inverse of the MTFA. However, the MTFA is a good measure of the FAR if, and only if, the pre-change distributions of the WL-CuSum stopping time $\tilde{\tau}_C(b)$ and the WL-GLR-CuSum stopping time $\tilde{\tau}_G(b)$ are approximately geometric. 
While this geometric property can be established for i.i.d.\ data models (see, e.g., Pollak and Tartakovsky~\cite{PollakTartakovskyTPA09} and Yakir~\cite{Yakir-AS95}), it is not neccessarily true for non-homogeneous and dependent data, as discussed in Mei~\cite{Mei-SQA08} and Tartakovsky~\cite{Tartakovsky-SQA08a}. 
Therefore, in general, the MTFA is not appropriate for measuring the FAR. In fact, large values of MTFA may not necessarily guarantee small values of the probability of false alarm as discussed in detail in \cite{Tartakovsky-SQA08a,tartakovsky_sequential}. When the post-change model is Gaussian non-stationary as defined in Example~\ref{gex}, the MTFA 
may still be an appropriate measure for false alarm rate, as shown in the simulation study in Section~\ref{num-res:mtfa}. Based on this result we conjecture that the MTFA-based FAR constraint may be suitable for other independent and non-stationary data models as well. However, in general, this may not be the case, and 
a more appropriate measure of the FAR in the general case may be the maximal (local) conditional probability of false alarm in the time interval $(k, k+m]$ defined as \cite{tartakovsky_sequential}:
\[
    \SPFA_m(\tau) = \sup_{k \ge 0} \Prob{\infty}{\tau \le k+m | \tau > k}.
\]
Then the constraint set in \eqref{fa_constraint} can be replaced by set
$
\Cbb_{\beta,m} = \{\tau: \SPFA_m(\tau) \leq \beta\}
$
of procedures for which the SPFA does not exceed a prespecified value $\beta \in (0,1)$. 

Pergamenschtchikov and Tartakovsky~\cite{PergTarSISP2016,PerTar-JMVA2019} considered general stochastic models of dependent and 
nonidentically distributed observations
but asymptotically homogeneous (i.e., $\gfun(n)=n$). They proved not only minimax optimality but also asymptotic pointwise optimality as $\beta\to0$ 
(i.e., for all change points $\nu \ge 1$) of the Shiryaev-Roberts (SR)
procedure for the simple post-change hypothesis, and the mixture SR for the composite post-change hypothesis in class $\Cbb_{\beta,m}$, when 
$m=m_\beta$ depends on $\beta$ and goes to infinity as $\beta\to0$ at such a rate that $\log m_\beta=o(|\log \beta|)$. 

The results of 
\cite{PergTarSISP2016,PerTar-JMVA2019} can be readily extended to the asymptotically non-homogeneous case where the function $\gfun (n)$ increases with 
$n$ faster than $\log n$. In particular, using the developed in \cite{PergTarSISP2016,PerTar-JMVA2019} techniques based on embedding class $\Cbb_{\beta,m}$
in the Bayesian class with a geometric prior distribution for the change point and the upper-bounded weighted PFA,
 it can be shown that the WL-CuSum 
procedure \eqref{TC:test} with $m_\alpha$ replaced by $m_\beta$ is first-order pointwise asymptotically optimal in class 
$\Cbb_{\beta,m_\beta}=\Cbb_\beta$ as long as the uniform complete version 
of the strong law of large numbers for the log-likelihood ratio holds, i.e., for all $\delta > 0$
\[
    \sum_{n=1}^\infty \sup_{\nu \ge 1} \Prob{\nu}{\left | \frac{1}{\gfun_\nu(n)} \sum_{i=\nu}^{\nu+n-1} Z_{i, \nu} -1 \right | > \delta} < \infty,
\]
where in the general non-i.i.d. case
the partial LLR $Z_{i, \nu} $ is
\[
    Z_{i, \nu} = \log \frac{p_{1, i,\nu} (X_i| X_1,\dots, X_{i-1})}{p_{0, i}(X_i| X_1,\dots, X_{i-1})} .
\]
Specifically, it can be established that for all fixed $\nu \ge 1$, as $\beta\to 0$,
\begin{align*}
    \inf_{\tau \in \Cbb_\beta} \ADD_\nu(\tau) \sim \ADD_\nu(\tilde{\tau}_C(b_\alpha)) \sim g^{-1}(|\log \beta|),  
\end{align*}
where we used the notation $\ADD_\nu(\tau)=\E{\nu}{\tau-\nu | \tau \geq \nu}$ for the conditional average delay to detection. Similar results also hold
for the maximal average detection delays $\WADD{\tau}$ and $\SADD{\tau}= \sup_{\nu \ge 1} \ADD_\nu(\tau)$.

It is worth noting that it follows from the proof of Theorem~\ref{llr:delay_lower_bound} that under condition \eqref{llr:upper} the following asymptotic lower bound holds for the average detection delay
$\ADD_\nu(\tau)$ uniformly for all values of the change point in class $\Cbb_\beta$:
\[
    \inf_{\tau\in\Cbb_\beta} \ADD_\nu(\tau) \ge g^{-1}(|\log \beta|) (1+o(1)), \quad \forall \nu \ge 1 ~ \text{as}~ \beta \to 0.
\]

In the case where the post-change observations have parametric uncertainty, sufficient conditions for the optimality of the WL-GLR-CuSum procedure are more sophisticated -- a probability in the vicinity of the true post-change parameter should be involved \cite{PerTar-JMVA2019}. 

Further details and the proofs are omitted and will be given elsewhere.

\section{Numerical Results}
\label{sec:num-res}

\subsection{Performance Analysis for GEM problem}
\label{num-res:pa}

\begin{figure}[tbp]
\centerline{\includegraphics[width=.74\textwidth,height=9cm]{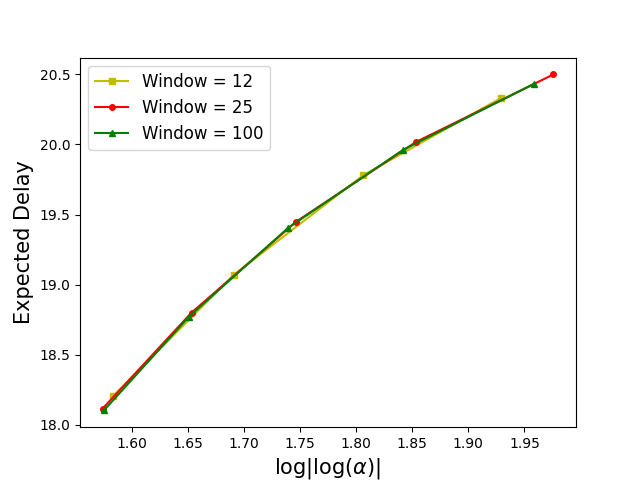}}
\vspace{-3mm}\caption{Performances of the WL-CuSum with different window-sizes for the Gaussian exponential mean-change detection problem  with $\mu_0=0.1$, $\sigma_0^2 = 10000$, and $\theta = 0.4$. The change-point is $\nu = 1$.}
\label{fig:perf}
\end{figure}

In Fig. \ref{fig:perf}, we study the performance of the proposed WL-CuSum procedure in \eqref{TC:test} through Monte Carlo (MC) simulations for the Gaussian exponential mean-change detection problem (see Example~\ref{gex}), with known post-change parameter $\theta$. The change-point is taken to be $\nu=1$\footnote{Note that $\nu=1$ may not necessarily be the worst-case value for the change-point for the WL-CuSum procedure. However, extensive experimentation with different values of $\nu$ ranging from 1 to 100, with window-sizes of 15 and 25, shows that in almost all cases $\nu=1$ results in the largest expected delay, or one that is within 1\% of the largest expected delay.}. Three  window-sizes are  considered,  with the window size of 12 being  smaller than the range expected delay values in the plot, and therefore not large enough to 
satisfy condition \eqref{TC:m_alpha}. The window size of 25 is sufficiently large, and the window size of 100 essentially corresponds to having no window at all. It is seen that the performance is nearly identical for all window sizes considered. We also observe that the expected delay is $O(\log(\abs{\log\alpha}))$, which matches our theoretical analysis in \eqref{gex:perf}.

\begin{figure}[tbp]
\centerline{\includegraphics[width=.75\textwidth,height=9cm]{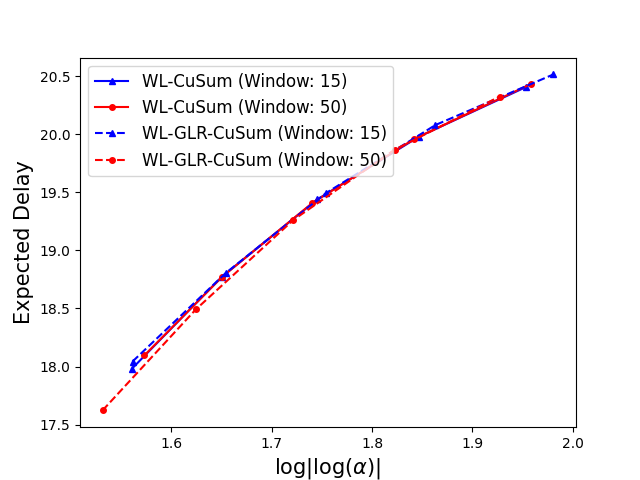}}
\vspace{-3mm}\caption{Comparison of operating characteristics of the WL-CuSum (solid lines) and WL-GLR-CuSum (dotted lines) procedures with different sizes of windows for the Gaussian exponential mean-change detection problem  with $\mu_0=0.1$, $\sigma_0^2 = 10000$, and $\theta = 0.4$. The post-change parameter set is $\Theta = (0,0.5)$, \vvv{which is further discretized into a grid with 50 equally spaced points for computing the GLR statistic.} The change-point $\nu = 1$. Procedures with sufficiently large (in red, circle) and insufficiently large window-sizes (in blue, triangle) are also compared. 
}
\label{fig:perf_glr}
\end{figure}

In Fig. \ref{fig:perf_glr}, we compare, also through MC simulations for the problem of Example~\ref{gex}, the performance of the WL-CuSum procedure \eqref{TC:test} tuned to the true post-change parameter and the WL-GLR-CuSum procedure \eqref{TG:test} where only the set of post-change parameter values is known. 
\vvv{It is seen that the operating characteristic of the WL-GLR-CuSum procedure is close to that of the WL-CuSum procedure for a sufficiently large window-size. We also observe that procedures with slightly insufficiently large window-sizes perform similarly to those with sufficiently large window sizes.}

\subsection{Performance Analysis for Gaussian Observations with Decaying Post-Change Mean}
\label{num-res:pa_dim}
In this subsection, we apply the WL-CuSum and WL-GLR-CuSum procedures for the QCD problem with  Gaussian observations, where the post-change mean gradually decays to the pre-change mean.
Specifically,
\begin{align}
\label{eq:gdm}
    X_n \sim {\cal N}(0,\sigma^2),~\forall n < \nu \nonumber\\
    X_n \sim {\cal N}(\mu_1 (n-\nu+1)^{-\theta},\sigma^2),~\forall n \geq \nu
\end{align}
for some decay parameter $\theta \in (0,1/2)$.
The growth function for this model is
\begin{equation*}
    g_{\nu}(n) = \sum_{i=\nu}^{\nu+n-1} \E{\nu}{Z_{i,\nu}} = \sum_{i=1}^{n} \frac{\mu_1^2}{2 \sigma^2} i^{-2\theta} = \frac{\mu_1^2}{2 \sigma^2 (1-2\theta)} n^{1-2 \theta} (1+o(1)),~\text{as}~n \to \infty
\end{equation*}
and thus
\begin{equation}
\label{grm:perf}
    g^{-1}(x) = (2 \sigma^2 \mu_1^{-2} (1-2\theta) x)^\frac{1}{1-2 \theta} (1+o(1)),~\text{as}~x \to \infty.
\end{equation}
Therefore, $\log g^{-1}(x) = O(\log x) = o(x)$ and condition \eqref{llr:grow_cond} is satisfied. Also note that since $(1-2\theta)^{-1} > 1$, the optimal WADD and CADD are asymptotically super-linear with $\abs{\log \alpha}$.

\begin{figure}[tbp]
\centerline{\includegraphics[width=.75\textwidth,height=9cm]{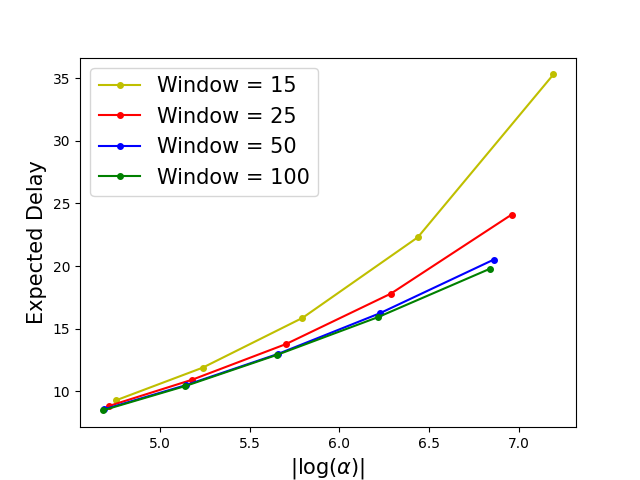}}
\vspace{-3mm}
\caption{Performances of the WL-CuSum with different window-sizes for the model in \eqref{eq:gdm} with $\mu_1=2$, $\sigma^2 = 4$, and $\theta = 0.2$. The change-point is $\nu = 1$.}
\label{fig:rec_cusum}
\end{figure}

In Fig.~\ref{fig:rec_cusum}, we study the performance of the proposed WL-CuSum procedure in \eqref{TC:test} through Monte Carlo (MC) simulations for the model in \eqref{eq:gdm}, with known decay parameter $\theta=0.35$. The change-point is taken to be $\nu=1$. Four window-sizes are considered, with the window sizes of 15 and 25 being smaller than the range of expected delay values in the plot, and therefore not large enough to 
satisfy condition \eqref{TC:m_alpha}.
It is seen that the performance improves significantly with an initial increase of window-size, with diminishing returns when the window-size become large enough.
We also observe that the expected delay is super-linear with $\abs{\log\alpha}$, which matches our theoretical analysis in \eqref{grm:perf}.

\begin{figure}[tbp]
\centerline{\includegraphics[width=.75\textwidth,height=9cm]{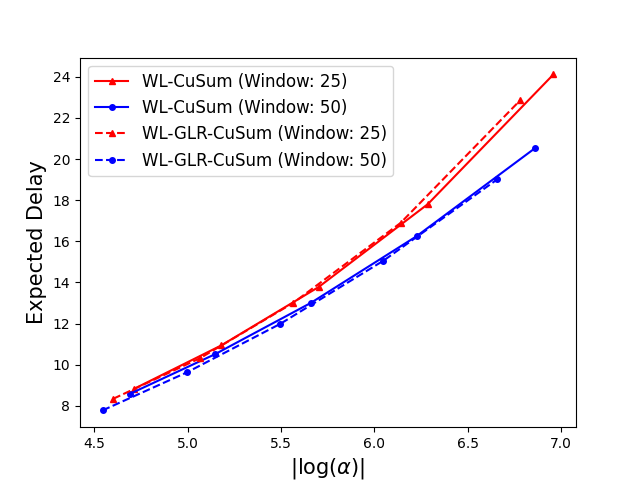}}
\vspace{-3mm}
\caption{Comparison of operating characteristics of the WL-CuSum (solid lines) and WL-GLR-CuSum (dotted lines) procedures with different sizes of windows for the observation model in  \eqref{eq:gdm} with $\mu_1=2$, $\sigma^2 = 4$, and $\theta = 0.2$. The post-change parameter set is $\Theta = (0.1,0.3)$, which is further discretized into a grid with 50 equally spaced points for computing the GLR statistic. The change-point $\nu = 1$.}
\label{fig:rec_glr}
\end{figure}

In Fig. \ref{fig:rec_glr}, we compare, also through MC simulations for the model in \eqref{eq:gdm}, the performance of the WL-CuSum procedure \eqref{TC:test} tuned to the true post-change parameter and the WL-GLR-CuSum procedure \eqref{TG:test} where only the \emph{set} of post-change parameter values is known.
We observe that the operating characteristic of WL-GLR-CuSum procedure is nearly identical to that of the WL-CuSum procedure for large enough window-size.

\subsection{Analysis of MTFA as False Alarm Measure}
\label{num-res:mtfa}

\begin{figure}[tbp]
\centerline{\includegraphics[width=.75\textwidth,height=10cm]{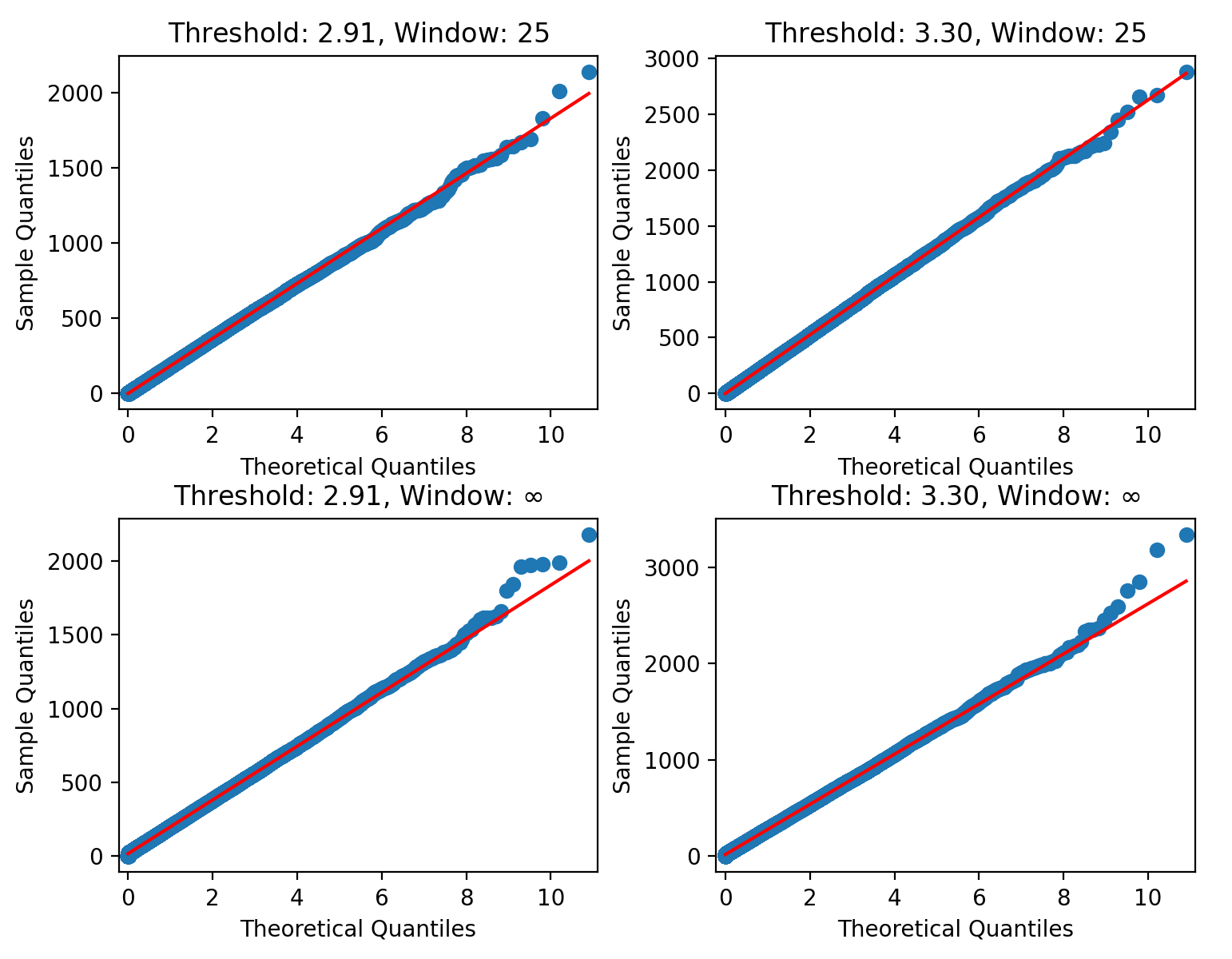}}
\vspace{-3mm}\caption{Quantile-Quantile (QQ) plots for full-history and window-limited CuSum stopping times with different thresholds for the Gaussian exponential mean-change detection problem with $\mu_0=0.1$, $\sigma_0^2 = 10000$, and $\theta = 0.4$. In all subplots, the x-axis shows the theoretical quantiles of the best-fit geometric distribution and the y-axis shows the experimental quantiles of distributions of the stopping times. The first row corresponds to WL-CuSum procedure  \eqref{TC:test} and the second row corresponds to the full-history CuSum procedure \eqref{TC:def}.}
\label{fig:qq}
\end{figure}

In Fig. \ref{fig:qq}, we study the distribution of the WL-CuSum stopping times using simulation results from the Gaussian exponential mean-change detection problem. This study is similar to the one in \cite{PollakTartakovskyTPA09}.
It is observed that the experimental quantiles of stopping times for the WL-CuSum procedure are close to the theoretical quantiles of a geometric distribution. This indicates that the distribution of the stopping time is approximately geometric, in which case MTFA is an appropriate false alarm performance measure, and our measure of FAR as the reciprocal of the MTFA is justified.

\subsection{Application: Monitoring COVID-19 Second Wave}
\label{num-res:covid}

\begin{figure}[htbp]
\centerline{\includegraphics[width=.75\textwidth,height=8cm]{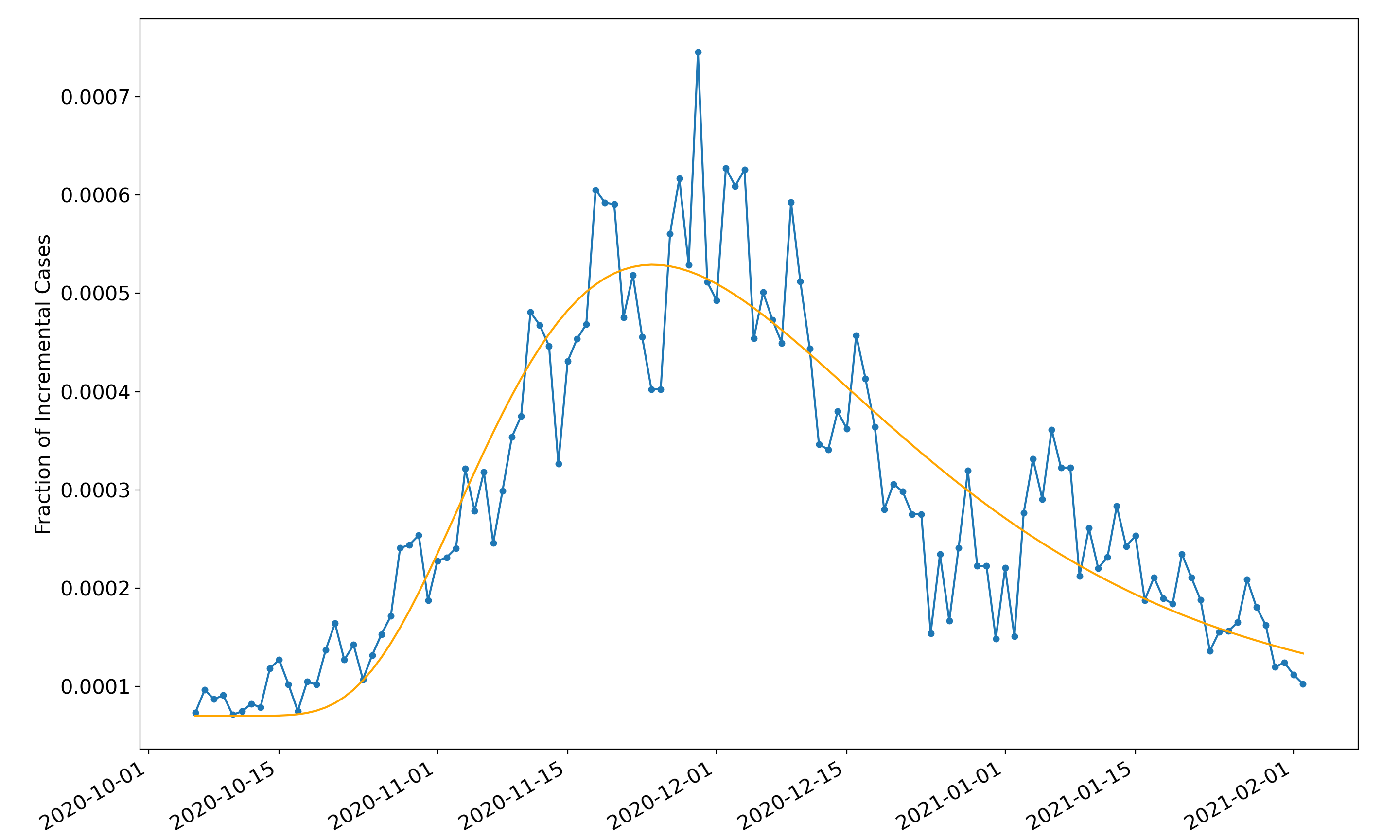}}
\vspace{-3mm}\caption{Validation of distribution model using past COVID-19 data.
The plot shows the four-day moving average of the daily new cases of COVID-19 as a fraction of the population in Wayne County, MI from October 1, 2020 to February 1, 2021 (in blue). The shape of the pre-change distribution $\mathcal{B}(a_0,b_0)$ is estimated using data from the previous 20 days (from September 11, 2020 to September 30, 2021), where $\hat{a}_0=20.6$ and $\hat{b}_0=2.94 \times 10^5$. The mean of the Beta distributions with the best-fit $h$ (defined in \eqref{num_sim:h}) is also shown (in orange), which minimizes the mean-square distance between the daily incremental fraction and mean of the Beta distributions. The best-fit parameters are: $\hat{\theta}_0=0.464$, $\hat{\theta}_1=3.894$, and $\hat{\theta}_2=0.445$.
}
\label{fig:cov_val}
\end{figure}

Next, we apply the developed WL-GLR-CuSum algorithm to monitoring the spread of COVID-19 using new case data from various counties in the US \cite{nyt-covid-data}. The goal is to detect the onset of a new wave of the pandemic based on the incremental daily cases. 
The problem is modeled as one of detecting a change in the mean of a Beta distribution as in \cite{covid_beta}. Let $\mathcal{B}(x;a,b)$ denote the density of the Beta distribution with shape parameters $a$ and $b$, i.e.,
\begin{equation*}
    \mathcal{B}(x;a,b) = \frac{x^{a-1}(1-x)^{b-1}\Gamma(a+b)}{\Gamma(a)\Gamma(b)}, \quad \forall x \in [0,1],
\end{equation*}
where $\Gamma$ represents the gamma function. Note that the mean of an observation under density $\mathcal{B}(x;a,b)$ is $a / (a+b)$.
Let
\begin{equation}
\label{num_sim:dist_model}
    p_0(x) = \mathcal{B}(x;a_0,b_0), \quad p_{1,n,k}^\theta (x) = \mathcal{B}(x;a_0 h_\theta(n-k),b_0), ~\forall n \geq k.
\end{equation}
Here, $h_\theta$ is a function such that $h_\theta(x) \geq 1,\forall x > 0$. Note that if $a_0 \ll b_0$ and $h_\theta(n-\nu)$ is not too large,
\begin{equation}
    \E{\nu}{X_n} = \frac{a_0 h_\theta(n-\nu)}{a_0 h_\theta(n-\nu) + b_0} \approx \frac{a_0}{b_0} h_\theta(n-\nu)
\end{equation}
for all $n \geq \nu$.
We design $h_\theta$ to capture the behavior of the average fraction of daily incremental cases. In particular, we model $h_\theta$ as
\begin{equation}
\label{num_sim:h}
    h_\theta(x) = 1+\frac{10^{\theta_0}}{\theta_2} \exp\left(-\frac{(x-\theta_1)^2}{2 \theta_2^2} \right)
\end{equation}
where $\theta_0,\theta_1,\theta_2 \geq 0$ are the model parameters and $\theta = (\theta_0,\theta_1,\theta_2) \in \Theta$.  When $n-\nu$ is small, $h_\theta(n-\nu)$ grows like the left tail of a Gaussian density, which matches the exponential growth in the average fraction of daily incremental cases seen at the beginning of a new wave of the pandemic. Also, as $n \to \infty$, $h_\theta(n-\nu) \to 0$, which corresponds to the daily incremental cases eventually vanishing at the end of the pandemic. In Fig. \ref{fig:cov_val}, we validate the choice of distribution model defined in \eqref{num_sim:dist_model} using data from COVID-19 wave of Fall 2020. 
In the simulation, $a_0$ and $b_0$ are estimated using observations from previous periods in which the increments remain low and roughly constant. It is observed that the mean of the daily fraction of incremental cases matches well with the mean of the fitted Beta distribution with $h_\theta$ in \eqref{num_sim:h}.

Note that the growth condition given in \eqref{TG:growth_cond} that is required for our asymptotic analysis is not satisfied for the observation model \eqref{num_sim:dist_model} with $h_\theta$ given in \eqref{num_sim:h}. Nevertheless, we expect the WL-GLR-CuSum procedure to perform as predicted by our analysis if the procedure stops during a time interval where $h_\theta$ is still increasing, which is what we would require of a useful procedure for detecting the onset of a new wave of the pandemic anyway.

\begin{figure}[htbp]
\centerline{\includegraphics[width=\textwidth,height=10cm]{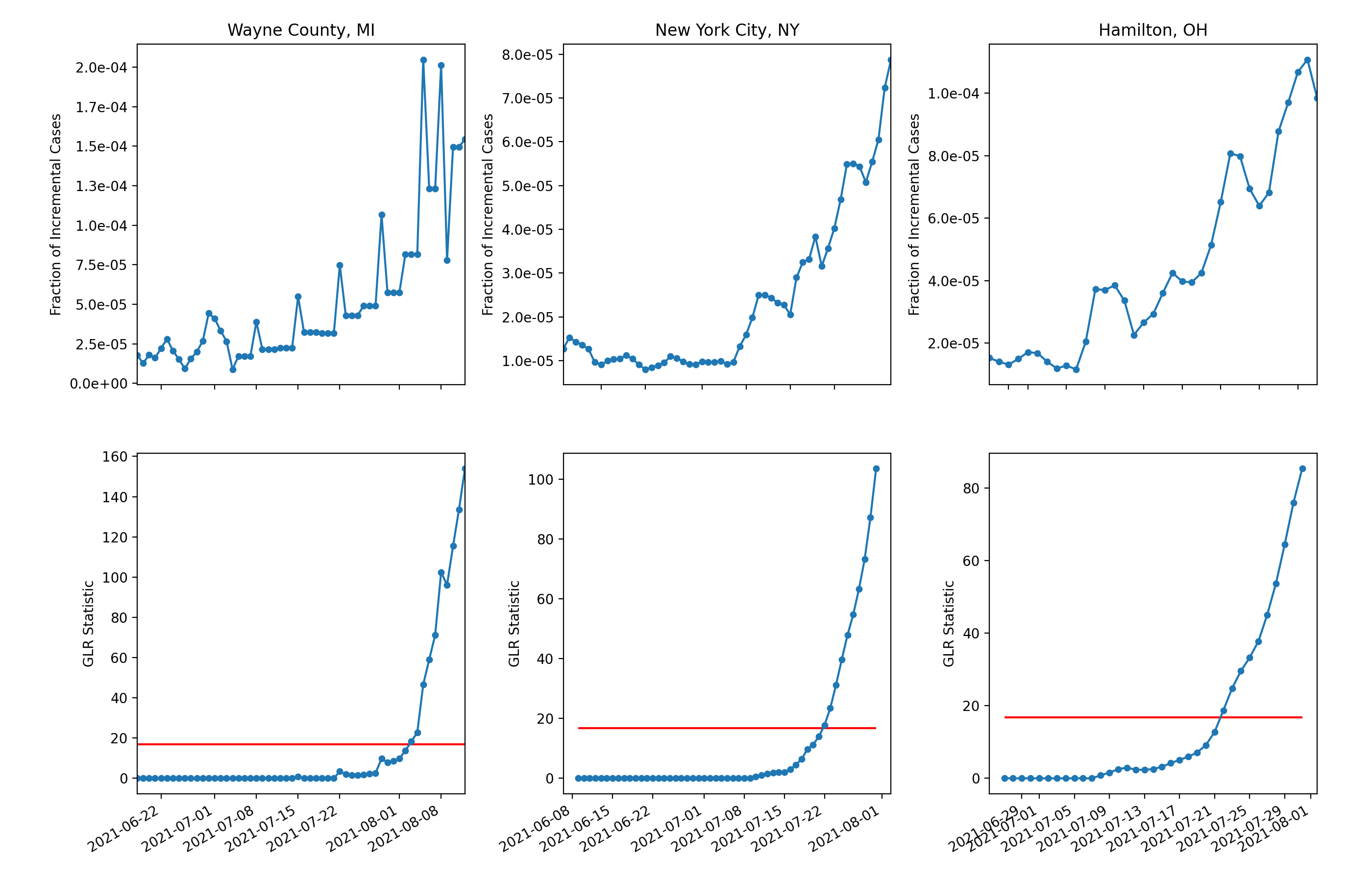}}
\vspace{-3mm}\caption{COVID-19 monitoring example. The upper row shows the four-day moving average of the daily new cases of COVID-19 as a fraction of the population in Wayne County, MI (left), New York City, NY (middle) and Hamilton County, OH (right). A pre-change $\mathcal{B}(a_0,b_0)$ distribution is estimated using data from the previous 20 days (from May 26, 2021 to June 14, 2021). 
The plots in the lower row show the evolution of the WL-GLR-CuSum statistic defined in \eqref{TG:test}. The FAR $\alpha$ is set to $0.001$ and the corresponding thresholds of the WL-CuSum GLR procedure are shown in red. The post-change distribution at time $n$ with hypothesized change point $k$ is modeled as $\mathcal{B}(a_0 h_\theta(n-k),b_0)$, where $h_\theta$ is defined in \eqref{num_sim:h}, and $\Theta = (0.1,5) \times (1,20) \times (0.1,5)$. The parameters $\theta_0$, $\theta_1$ and $\theta_2$ are assumed to be unknown. The window size $m_\alpha = 20$. The threshold is set using equation \eqref{TG:thr}.}
\label{fig:covid}
\end{figure}

In Fig. \ref{fig:covid}, we illustrate the use the WL-GLR-CuSum procedure with the distribution model \eqref{num_sim:dist_model} for the detection of the onset of a new wave of COVID-19. We assumed a start date of June 15th, 2021 for the monitoring, at which time the pandemic appeared to be in a steady state with incremental cases staying relatively flat. We observe that the WL-GLR-CuSum statistic significantly and persistently crosses the detection threshold around late July in all counties, which is strong indication of a new wave of the pandemic. More importantly, unlike the raw observations which are highly varying, the WL-GLR-CuSum statistic shows a clear dichotomy between the pre- and post-change settings, with the statistic staying near zero before the purported onset of the new wave, and taking off very rapidly (nearly vertically) after the onset. 

\section{Conclusion}
\label{sec:concl}

We considered the problem of the quickest detection of a change in the distribution of a sequence of independent observations, assuming that the pre-change observation are stationary with known distribution, while the post-change observations are non-stationary with possible parametric uncertainty. Specifically, we assumed that the cumulative KL divergence between the post-change and the pre-change distributions grows at least logarithmically after the change point. We derived a universal asymptotic lower bound on the worst-case expected detection delay under a constraint on the false alarm rate in this non-stationary setting, which had been previously derived only in the asymptotically stationary setting. We showed that the developed WL-CuSum procedure for known post-change distribution, as well as the developed WL-GLR-CuSum  procedure for the unknown post-change parameters, asymptotically achieve the lower bound on the worst-case expected detection delay, as the false alarm rate goes to zero. We validated these theoretical results through numerical Monte-Carlo simulations. We also demonstrated that the proposed WL-GLR-CuSum procedure can be effectively used in monitoring pandemics. We provided in Section~\ref{sec:ext} some possible avenues for future research, in particular, those allowing for dependent observations and more general false alarm constraints. 

\section{Acknowledgment}
\label{sec:ack}

The authors would like to thank Don Towsley for suggesting the example in \eqref{eq:gdm}.

\appendix

\begin{proof}[\textbf{Proof of Lemma~\ref{llr:lemma}}]
For the first inequality, fix $\nu \geq 1$ and $\delta > 0$. Note that
\[ \gfun_\nu (n) = \sum_{i=\nu}^{\nu+n-1} \E{\nu}{Z_{i,\nu}}. \]
Since $\E{\nu}{Z_{i,\nu}} > 0$ for all $i \geq \nu$, for any $t \leq n$,
\begin{equation*}
    \gfun_\nu(n) = \sum_{i=\nu}^{\nu+n-1} \E{\nu}{Z_{i,\nu}} \geq \gfun_\nu(t)
\end{equation*}
and, by definition,
\begin{equation*}
    \E{\nu}{\sum_{i=\nu}^{\nu+t-1} Z_{i,\nu} - \gfun_\nu(t)} = 0.
\end{equation*}
Thus, for an arbitrary $\nu > 1$ we have
\begin{align}
    \Prob{\nu}{\max_{t \leq n} \sum_{i=\nu}^{\nu+t-1} Z_{i,\nu} \geq (1+\delta) \gfun_\nu(n)} &\leq \Prob{\nu}{\max_{t \leq n} \sum_{i=\nu}^{\nu+t-1} Z_{i,\nu} - \gfun_\nu(t) \geq \delta \gfun_\nu(n)} \nonumber\\
    &\leq \Prob{\nu}{\max_{t \leq n} \abs{\sum_{i=\nu}^{\nu+t-1} Z_{i,\nu} - \gfun_\nu(t)} \geq \delta \gfun_\nu(n)} \nonumber\\
    &\stackrel{(*)}{\leq} \frac{1}{\delta^2 \gfun_\nu^2(n)} \Var{\nu}{\sum_{i=\nu}^{\nu+n-1} Z_{i,\nu} - \gfun_\nu(n)} \nonumber\\
    &= \frac{1}{\delta^2 \gfun_\nu^2(n)} \sum_{i=\nu}^{\nu+n-1} \Var{\nu}{Z_{i,\nu}} 
\end{align}
where $(*)$ follows from Kolmogorov's inequality and the last line follows by independence.
Hence, 
\begin{align*}
    \sup_{\nu \ge 1} \Prob{\nu}{\max_{t \leq n} \sum_{i=\nu}^{\nu+t-1} Z_{i,\nu} \geq (1+\delta) \gfun_\nu(n)} \leq
     \frac{1}{\delta^2} \sup_{\nu \ge 1} \frac{1}{\gfun_\nu^2(n)} \sum_{i=\nu}^{\nu+n-1} \Var{\nu}{Z_{i,\nu}} \xrightarrow{n \to \infty} 0
\end{align*}
where the limit follows from condition \eqref{llr:var}.

For the second inequality, fix $\nu$ and $t$ such that $t \geq \nu \geq 1$.
For any $\delta \in (0,1)$, we have
\begin{align}
    &\quad \Prob{\nu}{\sum_{i=t}^{t+n-1} Z_{i,t} \leq (1-\delta) \gfun_\nu(n)} \nonumber\\
    &= \Prob{\nu}{\sum_{i=t}^{t+n-1} (Z_{i,t} - \E{\nu}{Z_{i,t}}) \leq -\delta \gfun_\nu(n) + \underbrace{\sum_{i=\nu}^{\nu+n-1} \E{\nu}{Z_{i,\nu}} - \sum_{i=t}^{t+n-1} \E{\nu}{Z_{i,t}}}_{\leq 0 \text{ by \eqref{llr:tshift}}} } \nonumber\\
    &\leq \Prob{\nu}{\sum_{i=t}^{t+n-1} (Z_{i,t} - \E{\nu}{Z_{i,t}}) \leq -\delta \gfun_\nu(n)} \nonumber\\
    &\leq \Prob{\nu}{\abs{\sum_{i=t}^{t+n-1} (Z_{i,t} - \E{\nu}{Z_{i,t}})} \geq \delta \gfun_\nu(n) } \nonumber\\
    &\stackrel{(**)}{\leq} \frac{1}{\delta^2 \gfun_\nu^2(n)} \sum_{i=t}^{t+n-1} \Var{\nu}{Z_{i,t}}
\end{align}
where $(**)$ follows from Chebyshev's inequality.
Thus,  by condition \eqref{llr:var},
\begin{equation*}
    \sup_{t \geq \nu \geq 1} \Prob{\nu}{\sum_{i=t}^{t+n-1} Z_{i,t} \leq (1-\delta) \gfun_\nu(n)} \le  \frac{1}{\delta^2} 
    \sup_{t \ge \nu \ge 1} \frac{1}{\gfun_\nu^2(n)} \sum_{i=t}^{t+n-1} \Var{\nu}{Z_{i,t}}  \xrightarrow{n \to \infty} 0.
\end{equation*}
The proof is complete. \qedhere
\end{proof}

\begin{proof}[\textbf{Proof of Lemma~\ref{TG:fa}}]

Let $\lambda_{n,k}^\theta = \sum_{i=k}^n Z_{i,k}^\theta$ denote the log-likelihood ratio between the hypotheses that $\nu=k$ with the parameter $\theta$ against $\nu=\infty$ in the sample $(X_1,\dots,X_n)$. We re-write the definition of $\td{\tau}_G(b)$ in \eqref{TG:test} as:
\begin{equation} \label{GRLCUSUMnew}
    \td{\tau}_G\left(b\right) := \inf \left\{n:\max_{n-m_b \leq k \leq n} \lambda_{n,k}^{\hat{\theta}_{n,k}} \geq b \right\}
\end{equation}
where $\hat{\theta}_{n,k} \in \Theta_b$ solves
\begin{equation}
\label{TGFA:sup_theta}
    \sup_{\theta \in \Theta_b} \sum_{i=k}^n Z_{i,k}^{\theta}
\end{equation}
for a given pair $(k,n)$ where $k \leq n$. Note that now instead of $\Theta_\alpha$ we use the notation $\Theta_b$, which is a compact subset of $\Theta$.

Let $\Pi$ be a probability measure on $\Theta_b$. Recall that ${\cal F}_n =\sigma(X_1,\dots,X_n)$ 
and ${\cal F}_\infty =\sigma(X_1,X_1,\dots)$. Given this mixing distribution over $\Theta_b$, define $Q_k$ by
\[
Q_k(A) = \int_{\Theta_b} \Pb_{k,\theta} (A) \Pi(\D \theta), \quad A \in {\cal F}_\infty.
\]
Then $Q_k$ is easily seen to be a probability measure. Moreover, letting $\Pb_{\infty}^n$ and $Q_k^n$ denote the restrictions of  $\Pb_{\infty}$ and $Q_k$
to the sigma-algebra ${\cal F}_n$ introduce the likelihood ratio
\[
    L_{n,k} := \frac{\D Q_k^n}{\D \Pb_{\infty}^n} = \int_{\Theta_b} \exp(\lambda_{n,k}^\theta) \Pi(\D\theta), \quad n \ge k.
\]
Further, let
\[
    R_n := \sum_{k=1}^n L_{n,k}.
\]
Obviously, $\E{\infty}{L_{n,k}|{\cal F}_{n-1}} = L_{n-1,k}$ and $\E{\infty}{L_{n,k}}=1$, so $\{L_{n,k}\}_{n \ge k}$ is a 
$(\mathbb{P}_\infty,{\cal F}_n)$-martingale with unit expectation. Hence, $\E{\infty}{R_n | {\cal F}_{n-1}} = 1 + R_{n-1}$ and the statistic $\{R_n-n\}_{n \ge 1}$ 
is a zero-mean $(\mathbb{P}_\infty,{\cal F}_n)$-martingale. By the optional sampling theorem (see, e.g., \cite[Th 2.3.1, page 31]{tartakovsky_sequential}),
for any proper stopping time $\tau$, $\E{\infty}{R_\tau} = \E{\infty}{\tau}$, and in particular, $\E{\infty}{R_{\td{\tau}_G(b)}} = \E{\infty}{\td{\tau}_G(b)}$.   
 
Now, set $\Pi(\D\theta) = |\Theta_b|^{-1} \D \theta$ (uniform on $\Theta_b$). At the next step we show that as $b \to \infty$
\begin{equation}\label{RtoV}
    \E{\infty}{R_{\td{\tau}_G(b)}} \ge |\Theta_b|^{-1} e^{-1} \frac{\pi^{d/2}}{\Gamma(1+d/2)} b^{- \frac{\eps d}{2}} e^b (1+o(1)),
\end{equation}
which along with the previous argument implies that
\begin{align*}
\E{\infty}{\td{\tau}_G(b)}
& \ge |\Theta_b|^{-1} e^{-1} \frac{\pi^{d/2}}{\Gamma(1+d/2)} b^{- \frac{\eps d}{2}} e^b (1+o(1)).
\end{align*}
This inequality implies inequality \eqref{FARGLR}. Thus, it remains to prove the asymptotic inequality \eqref{RtoV}.

By assumption, $\hat{\theta}_{n,k}$ lies in the interior of $\Theta_b$ (for sufficiently large $b$). 
Using Taylor's expansion, for any $k \leq n$ and $\theta \in \Theta$,
\begin{align}
\label{TGFA:Taylor}
    \lambda_{n,k}^{\theta} &= \lambda_{n,k}^{\hat{\theta}_{n,k}} + \left. \nabla_\theta \lambda_{n,k}^{\theta} \right|_{\theta = \hat{\theta}_{n,k}} (\theta-\hat{\theta}_{n,k}) + \frac{1}{2} (\theta-\hat{\theta}_{n,k})^\T \left. \nabla_\theta^2 \lambda_{n,k}^{\theta} \right|_{\theta = \theta^*} (\theta-\hat{\theta}_{n,k}) \nonumber\\
    &= \lambda_{n,k}^{\hat{\theta}_{n,k}} + \frac{1}{2} (\theta-\hat{\theta}_{n,k})^\T \left. \nabla_\theta^2 \lambda_{n,k}^{\theta} \right|_{\theta = \theta_{n,k}^*} (\theta-\hat{\theta}_{n,k})
\end{align}
where $\theta_{n,k}^*$ is an intermediate point between $\theta$ and $\hat{\theta}_{n,k}$, i.e.,  $\theta_{n,k}^* = \rho \theta + (1-\rho) \hat{\theta}_{n,k}$ 
for some $\rho \in (0,1)$. The last equality follows from \eqref{TGFA:sup_theta}. This further implies that
\begin{equation}
    \lambda_{n,k}^{\theta} - \lambda_{n,k}^{\hat{\theta}_{n,k}} \geq - \frac{1}{2} \emax{-\nabla_\theta^2 \lambda_{n,k}^{\theta_{n,k}^*}} \norm{\theta-\hat{\theta}_{n,k}}^2 \geq - \frac{1}{2} \Lambda_{n,k} \norm{\theta-\hat{\theta}_{n,k}}^2,
\end{equation}
where
\[
    \Lambda_{n,k} := \sup_{\theta: \norm{\theta-\hat{\theta}_{n,k}} < b^{-\frac{\eps}{2}}} \emax{- \nabla_\theta^2 \sum_{i=k}^n Z_{i,k}^{\theta}}.
\]

Fix $\eps > 0$ such that Assumption~\ref{TG:smooth} is satisfied. Write $S(b) := \left\{\theta: \norm{\theta-\hat{\theta}_{n,k}} < b^{-\frac{\eps}{2}}\right\}$.
Since $S(b) \searrow \varnothing$ while $\Theta_b \nearrow \Theta$  as $b \nearrow \infty$, it follows that $S(b) \subset \Theta_b$ for all sufficiently large $b$.

Denote $b_0 := \inf\{b>0:S(b) \subseteq \Theta_b \}$. For any $b > b_0$, we have
\begin{align}
   L_{n,k}e^{-\lambda_{n,k}^{\hat{\theta}_{n,k}}} &= |\Theta_b|^{-1} \int_{\Theta_b} \exp\left(\lambda_{n,k}^{\theta} - \lambda_{n,k}^{\hat{\theta}_{n,k}}\right) \D \theta\nonumber\\
    &\geq |\Theta_b|^{-1} \int_{S(b)} \exp\left(\lambda_{n,k}^{\theta} - \lambda_{n,k}^{\hat{\theta}_{n,k}}\right) \D \theta \nonumber\\
    &\geq |\Theta_b|^{-1} \int_{S(b)} \exp\left(-\frac{1}{2} \emax{-\nabla_\theta^2 \lambda_{n,k}^{\theta_{n,k}^*}} \norm{\theta-\hat{\theta}_{n,k}}^2\right) \D \theta \nonumber\\
    &\geq |\Theta_b|^{-1} \int_{S(b)} \exp\left(-\frac{1}{2} \Lambda_{n,k} \norm{\theta-\hat{\theta}_{n,k}}^2\right) \D \theta \nonumber\\
    &\geq |\Theta_b|^{-1} \exp\left(-\frac{1}{2} \Lambda_{n,k} b^{-\eps}\right) C_d b^{- \frac{\eps d}{2}} ,
\end{align}
where $C_d := \frac{\pi^{d/2}}{\Gamma(1+d/2)}$.
The last inequality follows because the volume of a $d$-dimensional ball with radius $b^{-\frac{\eps}{2}}$ is given by 
$\frac{\pi^{d/2}}{\Gamma(1+d/2)} b^{- \frac{\eps d}{2}} =: C_d b^{- \frac{\eps d}{2}}$, where $\Gamma(\cdot)$ is the gamma function. Therefore,
\[
    L_{n,k} \ge \exp\left(\lambda_{n,k}^{\hat{\theta}_{n,k}}\right) |\Theta_b|^{-1} \exp\left(-\frac{1}{2} \Lambda_{n,k} b^{-\eps}\right) C_d b^{- \frac{\eps d}{2}}.
\]
Write
\[
    V_n = \max_{n-m_b \le k \le n} \exp\left(\lambda_{n,k}^{\hat{\theta}_{n,k}}\right)
\]
and note that $V_{\td{\tau}_G(b)}\ge e^b$ on $\{\td{\tau}_G(b) < \infty\}$ (by the definition of the stopping time $\td{\tau}_G(b)$
in \eqref{GRLCUSUMnew}).
It follows that
\begin{align}
    \E{\infty}{R_{\td{\tau}_G(b)}} &= \E{\infty}{\sum_{k=1}^{\td{\tau}_G(b)} L_{\td{\tau}_G(b),k}} \nonumber
    \\
    &\ge \E{\infty}{\max_{1 \le k \le \td{\tau}_G(b)} \exp \left(\lambda_{\td{\tau}_G(b),k}^{\hat{\theta}_{\td{\tau}_G(b),k}}\right) \exp\left(-\frac{1}{2} \Lambda_{\td{\tau}_G(b),k} b^{-\eps}\right)} |\Theta_b|^{-1} C_d b^{- \frac{\eps d}{2}} \nonumber
    \\
    &\ge \E{\infty}{V_{\td{\tau}_G(b)} \min_{\td{\tau}_G(b) - m_b \le k \le \td{\tau}_G(b)} \exp\left(-\frac{1}{2} \Lambda_{\td{\tau}_G(b),k} b^{-\eps}\right)} |\Theta_b|^{-1} C_d b^{- \frac{\eps d}{2}} \nonumber
    \\
    &\ge \E{\infty}{\exp\left(-\frac{b^{-\eps}}{2} \max_{\td{\tau}_G(b) - m_b \le k \le \td{\tau}_G(b)} \Lambda_{\td{\tau}_G(b),k}\right)} e^b |\Theta_b|^{-1} C_d b^{- \frac{\eps d}{2}}.
    \label{IneqER}
\end{align}
By Assumption~\ref{TG:smooth}, 
\[
   \lim_{b\to\infty}  \mathbb{P}_\infty \left\{\max_{n - m_b \le k \le n} \Lambda_{n,k} \leq 2 b^\eps \right\} =1,
\]
and therefore, as $b\to\infty$,
\begin{align*}
   \E{\infty}{\exp\left(-\frac{b^{-\eps}}{2} \max_{\td{\tau}_G(b) - m_b \le k \le \td{\tau}_G(b)} \Lambda_{\td{\tau}_G(b),k}\right)} =e^{-1} +o(1)
\end{align*}
which along with inequality \eqref{IneqER} yields inequality \eqref{RtoV}.
\end{proof}

\bibliographystyle{IEEEtran}
\bibliography{ref}

\end{document}